\documentclass{amsart}
\usepackage[utf8]{inputenc}

\title{Annotated Hypergraphs: Models and Applications}
\date{\today}

\usepackage{amsaddr}
\usepackage[bookmarks=false, colorlinks=true]{hyperref} 
\usepackage[utf8]{inputenc}
\usepackage[x11names]{xcolor}
\usepackage{subfiles}
\usepackage{amsmath}
\usepackage{amssymb}
\usepackage{amsthm}
\usepackage{amsfonts}
\usepackage{todonotes}
\usepackage[margin=1.5in]{geometry}
\usepackage{pc_math}
\usepackage{algorithm}
\usepackage[algo2e,ruled,norelsize]{algorithm2e}
\usepackage{csquotes}
\usepackage{bbm}
\usepackage{graphicx}
\usepackage{pdflscape}
\usepackage{tikz}
\usetikzlibrary{arrows}
\usepackage[titletoc,title]{appendix}
\usepackage{multirow}
\usepackage{kbordermatrix}
\usepackage{subcaption}
\usepackage{MnSymbol}

\usepackage{booktabs}

\definecolor{comment_purple}{HTML}{dec5ed}

\usepackage{cleveref}

\crefname{thm}{Theorem}{Theorems}
\crefname{lm}{Lemma}{Lemmas}
\crefname{clm}{Claim}{Claims}
\crefname{cor}{Corollary}{Corollaries}
\crefname{conj}{Conjecture}{Conjectures}
\crefname{appsec}{Appendix}{Appendices}
\crefname{dfn}{Definition}{Definitions}
\crefname{prop}{Proposition}{Propositions}

\makeatletter
\def\blfootnote{\gdef\@thefnmark{}\@footnotetext}
\makeatother

\begin{document}

\author{Philip Chodrow$^{\ast,\dagger}$}
\address{
         Operations Research Center \\ 
         Massachusetts Institute of Technology \\ 
         77 Massachusetts Avenue, Cambridge MA 02139, USA}
\email[Corresponding author]{pchodrow@mit.edu}

\author{Andrew Mellor$^\ast$}
\address{Mathematical Institute  \\
         University of Oxford \\
         Woodstock Road, Oxford, OX2 6GG, UK
        }
\email{mellor@maths.ox.ac.uk}

\maketitle


\begin{abstract}
    Hypergraphs offer a natural modeling language for studying polyadic interactions between sets of entities. 
    Many polyadic interactions are asymmetric, with nodes playing distinctive roles. 
    In an academic collaboration network, for example, the order of authors on a paper often reflects the nature of their contributions to the completed work. 
    To model these networks, we introduce \emph{annotated hypergraphs} as natural polyadic generalizations of directed graphs. 
    Annotated hypergraphs form a highly general framework for incorporating metadata into polyadic graph models.  
    To facilitate data analysis with annotated hypergraphs, we construct a role-aware configuration null model for these structures and prove an efficient Markov Chain Monte Carlo scheme for sampling from it. 
    We proceed to formulate several metrics and algorithms for the analysis of annotated hypergraphs. 
    Several of these, such as assortativity and modularity, naturally generalize dyadic counterparts. 
    Other metrics, such as local role densities, are unique to the setting of annotated hypergraphs. 
    We illustrate our techniques on six digital social networks, and present a detailed case-study of the Enron email data set. 
    
    \smallskip
    \noindent 
\end{abstract}

\section{Introduction}                     
    \blfootnote{$^\dagger$ Corresponding author}
    \blfootnote{$^\ast$ PC and AM contributed equally to this paper.}
    Many data sets of contemporary interest log interactions between sets of entities of varying size. 
    In collaborations between scholars, legislators, or actors, a single project may involve an arbitrary number of agents. 
    A single email links at least one sender to one or more receivers. 
    A given chemical reaction may require a large set of reagents. 
    Networks such as these cannot be represented via the classical paradigm of dyadic graphs without loss of higher-order information. 
    For this reason, recent scholarly attention has emphasized the role of such polyadic interactions in governing the structure and function of complex systems \cite{Benson2018,greening2015higher,de2019social}.  
    Polyadic data representations such as hypergraphs \cite{berge1984hypergraphs,chodrow2019configuration} and simplicial complexes \cite{Young2017,carlsson2009topology} have emerged as practical modeling frameworks that directly represent interactions between arbitrary sets of agents. 

    In some cases, even polyadic data representations may be inadequate. 
    Increasingly, network data sets incorporate rich metadata over and above topological structure. 
    Models that flexibly incorporate this information can assist analysts in discovering features that may not be apparent without metadata. 
    In this article, we consider an important and relatively general class of metadata in which nodes are assigned \emph{roles} in each edge. 
    A variety of social data sets involve such roles. 
    For example, research articles have junior and senior authors.
    Political bills have sponsors and supporters. 
    Movies have starring and supporting actors. 
    Emails have senders, receivers, and carbon copies. 
    Chemical reactions possess reactants, solvents, catalysts and inhibitors.
    These roles induce asymmetries in edges, and permuting role labels within an edge results in a meaningfully different data set. 
    For example, a movie in which actor $A$ plays a starring role and $B$ a supporting role becomes a different movie if the roles are exchanged. 
    
    Metadata, including roles, can be especially important for modeling processes evolving on network substrates. 
    A trivial example is that information cannot flow along an email edge from receiver to sender. 
    A less trivial example comes from a recent study \cite{rotabi2017tracing} which found that conventions in scholarly documentation preparation tend to flow along collaborations from more senior authors to more junior ones. 
    In many fields, senior authors will tend to be ``last'' authors, while junior ones are more likely to be ``first'' or ``middle'' authors. 
    The author order therefore carries important information about the spread of conventions along this collaboration network. 

    There exists an extensive literature studying graphs and hypergraphs with metadata attached to nodes  \cite{ghoshal2009random, mcmorris1994triangulating,kovanen2013temporal, henderson2012rolx,peel2017ground} and edges \cite{mucha2010community,gomez2013diffusion,battiston2014structural}.
    The problem of studying hypergraphs with general roles, however, does not neatly fit into any of these frameworks. 
    This is because roles are not attributes of either nodes or edges, but rather of node-edge pairs. 
    An actress is not (intrinsically) a ``lead actress'' -- she may play a leading role in one film and a supporting role in the next. 
    Contextual metadata is familiar in the context of directed networks. 
    Each edge contains two nodes, one of which possesses the role ``source'' and the other ``target,'' however a node may be the source of some edges, and the target of others.
    In  \cite{gallo1993directed,gallo1998directed}, the authors allow edges to contain arbitrary numbers of nodes, each of which is assigned one of these two roles. 
    This results in \emph{directed hypergraphs}, which have found some application in  the study of cellular networks \cite{klamt2009hypergraphs} and routing \cite{marcotte1998hyperpath} problems.
    Subsequent work generalized further to ``multimodal networks'' by introducing a relationship of ``association'' \cite{heath2009multimodal} alongside the source and target roles.

    Many of these developments have a somewhat \emph{ad hoc} character, motivated by important modeling needs within specific application domains. 
    Our aim in this work is to develop a unified modeling and analysis framework for polyadic data with contextual roles, which can then be flexibly deployed in varied domains. 
    The article is structured as follows. 
    In \Cref{sec:definitions}, we define \emph{annotated hypergraphs}, which naturally generalize the notion of directedness to polyadic data. 
    We then define a configuration model for annotated hypergraphs, and prove a Markov Chain Monte Carlo algorithm for sampling from this model. 
    In \Cref{sec:observables}, we define a range of role-aware metrics for studying the structure of annotated hypergraphs. 
    Some of these are direct generalizations of familiar tools, including centrality, assortativity, and modularity. 
    Others, such as local role densities, are qualitatively novel.  
    In \Cref{sec:results}, we bring our methods to bear on a small collection of social network data sets, showing how the framework of annotated hypergraphs allows us to flexibly highlight interpretable features in the data. 
    Additionally, we conduct an extended case-study of the popular Enron email data set. 
    We conclude in \Cref{sec:discussion} with a discussion of our results and suggestions for future work in the modeling of rich, polyadic data. 
    

\section{Annotated Hypergraphs} \label{sec:definitions}

    \begin{dfn}[Annotated Hypergraph]
    \label{def:annotated_hypergraph}
    An \emph{annotated hypergraph} $\h = (\V, \E, \X, \ell)$ consists of: 
    \begin{enumerate}
        \item A node set $\V$.
        \item A labeled edge set $\E$, a multiset of subsets of $\V$. 
        In particular, multi-edges are permitted, but edges in which the same node appears twice are not.
        \item A finite label set $\X$. 
        \item A role labeling function $\ell:\{(v,e)\in \V \times \E | v \in e \} \rightarrow \X$ where $\X$ is a label alphabet. 
    \end{enumerate}
    \end{dfn}
    The statement $\ell(v, e) = x$ is to be read as ``node $v$ has role $x$ in edge $e$.'' 
    We emphasize that the labeling function is contextual.
    Roles are assigned neither to nodes or to edges, but rather to node-edge pairs.
    There are two representations of annotated hypergraphs that will be useful in our subsequent development.
    Let $n = \abs{\V}$, $m = \abs{\E}$, and $p = \abs{\X}$. 
    Let $\mathbb{H}$ refer to the set of all annotated hypergraphs with $n$ nodes, $m$ hyperedges, and label alphabet $\X$. 
    
    \begin{dfn}[Labeled Incidence Array]
        The \emph{labeled incidence array} $\mathbf{T} = \mathbf{T}(\mathcal{H})\in \{0,1\}^{n\times m \times p}$ of an annotated hypergraph is defined entrywise by 
        \begin{align}
            t_{vex} = \begin{cases}
                1 &\quad v \in e \text{ and } \ell(v,e) = x \\ 
                0 &\quad \text{otherwise.}
                           \end{cases}
        \end{align}
    \end{dfn}
    Note that the labeled incidence array is distinct from the tensorial representation of directed hypergraphs in \cite{xie2016spectral}.
    An alternative representation is especially useful in the design of algorithms. 
    \begin{dfn}[Annotated Bipartite Graph]
        The \emph{annotated bipartite graph} $\B(\h) = (\V', \E')$ consists of
        \begin{itemize}
            \item A node set $\V' = (\V, \E)$
            \item An edge set $\E'$, where an edge between $v$ and $e$ exists in $\E'$ iff $v\in e$ in $\E$. 
            Each edge is labeled by $\ell(v,e)$, and may therefore be written $(v,e;x)$. 
        \end{itemize}
        Let $\mathbb{B}$ be the set of annotated bipartite graphs. 
    \end{dfn}
    
    We generalize the notion of self-loops in graphs via the concept of \emph{degeneracy}. 
    \begin{dfn}[Degeneracy]
        An edge $e \in \E$ is \emph{role-degenerate} if the same node $v$ appears twice in $e$, with the same role. 
        Edge $e$ is simply \emph{degenerate} if the same node appears twice in $e$, possibly in different roles. 
        We call an annotated hypergraph $\h$ role-degenerate (resp. degenerate) if it contains any role-degenerate (resp. degenerate) edges. 
    \end{dfn}
    
    It is difficult to find a modeling justification for hypergraphs with role-degenerate edges, but degenerate edges may have modeling applications. 
    For example, in email networks, it may be useful to log instances in which a sender also copies themsleves into the recipients list. 
    In our experiments below, however, we work only on data with neither form of degeneracy.
    Throughout the remainder of the paper, we restrict $\mathbb{H}$ to the set of \emph{nondegenerate} annotated hypergraphs. 

\subsection{A Configuration Model for Annotated Hypergraphs}

    We now define a configuration null model on nondegenerate, annotated hypergraphs. 
    As has recently been  emphasized by \cite{Fosdick2018}, there are several distinct models that are often called ``configuration models,'' and care is  required in order to define and sample from them. 

    To do so, we define two sets of vectors that summarize the incidence array $\mathbf{T}$. 
    For each $x$, define the vector $\mathbf{d}^x \in \mathbb{Z}_+^{n}$ entrywise as 
    \begin{align}
    \label{eqn:node_degrees}
        d_v^x =  \sum_{e \in E} t_{vex}\;.
    \end{align}
    Similarly, define the vector $\mathbf{k}\in \mathbb{Z}_+^m$ entrywise as 
    \begin{align*}
        k_e^x =  \sum_{v \in V} t_{vex}\;.
    \end{align*}
    The vector $\mathbf{d}^x$ counts the number of times that each node plays role $x$ in $\h$, while the vector $\mathbf{k}^x$ counts the number of nodes with role $x$ in each edge.  
    Let $\mathbf{D}$ be the matrix whose $x$th column is $\mathbf{d}^x$, and $\mathbf{K}$ the matrix whose $x$th column is $\mathbf{k}^x$. 
    In a slight abuse of notation, we also regard $\mathbf{D}$ and $\mathbf{K}$ as functions of $\h$. 

    \begin{dfn}[Configuration Null Space]
    	The \emph{configuration null space} $\mathbb{C}_{\mathbf{D}, \mathbf{K}} \subset \mathbb{H}$ induced by degree-role matrix $\mathbf{D}$ and dimension-role matrix $\mathbf{K}$ is 
    	\begin{align*}
    		\mathbb{C}_{\mathbf{D}, \mathbf{K}} = \left\{\h \in \mathbb{H} \;: \; \mathbf{D}(\h) = \mathbf{D}\;,\;  \mathbf{K}(\h) = \mathbf{K}\right\}\;.
    	\end{align*}
    	If $\mathbb{C}_{\mathbf{D}, \mathbf{K}} \neq \emptyset$, we say that $\mathbf{D}$ and $\mathbf{K}$ are \emph{configurable}.  
    \end{dfn}
    Throughout the remainder of this paper, we assume that $\mathbf{D}$ and $\mathbf{K}$ are configurable. 
    Note that this is always the case when $\mathbf{D}$ and $\mathbf{K}$ are extracted from an empirical data set. 
    
    In $\mathbb{C}_{\mathbf{D}, \mathbf{K}}$, each node ``remembers'' how many times it played role each $x$ and each edge remembers how many nodes playing role $x$ were contained in it. 
    Summing over $x$, we see that nodes remember their degrees and edges their dimensions. 
    The null space $\mathbb{C}_{\mathbf{D}, \mathbf{K}}$ thus generalizes the hypergraph configuration null space of \cite{chodrow2019configuration}. 

    A natural approach to defining a null model is to define the uniform measure on $\mathbb{C}(\h_0)$. 
    Such a definition is attractive from a theoretical standpoint, since this measure is also the entropy-maximizing measure on $\mathbb{H}$ subject to the degree and dimension constraints. 
    However, methods for sampling from uniform models suffer from computational issues related to counting edge multiplicities and rejection probabilities, often resulting in slow sampling \cite{Fosdick2018,chodrow2019moments}. 
    We therefore instead follow the traditional path of \cite{Bollobas1980,Molloy1998}, and others in formulating a configuration model as the output of a stub-matching algorithm. 
    
    A \emph{stub} is an indexed pair $(v, x, i)$ of a node and a role; the index $i$ serves simply to distinguish stubs.  
    For each such pair, $d_v^x$ counts the number of times that node $v$ has role $x$ in $\h_0$.  
    We collect all stubs in a single multiset: 
    \begin{align*}
        \Sigma = \biguplus_{x \in \X} \biguplus_{v \in V} \underbrace{\left\{(v,x, 1),\ldots,(v,x, d_v^x)\right\}}_{d_v^x \text{ copies}}\;.
    \end{align*}
    Stub-matching proceeds via the following algorithm, which we describe informally. 
    For each edge $e$:
    \begin{enumerate}
        \item For each role $x$, uniformly sample $k_e^x$ stubs with role $x$ from $\Sigma$, without replacement, and combine them via multiset union. Add the result to the edge set $\E$.
        \item Send $\Sigma \mapsto \Sigma \setminus e$. 
    \end{enumerate}
    The algorithm terminates when it is impossible to form the next edge.
    When  $\mathbf{D}$ and $\mathbf{K}$ are configurable, stub-matching terminates when $\Sigma = \emptyset$ and produces a partition generating a stub-labeled hypergraph with specified degree and dimension sequences.  
    
    By construction, the output of stub-matching is distributed according to the uniform measure $\mu_0$ on the set $\Sigma_{\mathbf{D},\mathbf{K}}$ of partitions of indexed stubs subject to the first moment constraints. 
    Let $g:\Sigma_{\mathbf{D},\mathbf{K}} \rightarrow \mathbb{C}_{\mathbf{D},\mathbf{K}}$ be the map that sends to each such partition its associated hypergraph. 
    \begin{dfn}[Configuration Model]
        The \emph{annotated hypergraph configuration model} is the measure $\mu(\h) = \mu_0(g^{-1}(\h)|\h\text{  is nondegenerate})$.     
    \end{dfn}
    The configuration model $\mu$ weights elements of $\mathbb{C}_{\mathbf{D},\mathbf{K}}$ according to their likelihood of being realized via stub-matching, conditional on nondegeneracy. 
    In this, it differs from the uniform distribution on $\mathbb{C}_{\mathbf{D},\mathbf{K}}$, since the same annotated hypergraph can be realized through multiple configurations. 
    Indeed, any permutation of stubs of the form $(v, x, 1), (v, x, 2)$ does not alter the image under $g$.
    Because of this, the configuration model tends to give greater probabilistic weight to annotated hypergraphs in which there are many parallel edges. 

    By definition, we can in principle sample from $\mu$ by repeatedly performing stub-matching until a nondegenerate configuration is obtained, and then applying the function $g$. 
    This approach is usually impractical, as the probability of realizing a nondegenerate configuration is typically very low. 
    The generalization of limit laws such as those provided by \cite{Angel2016} for a dyadic configuration model governing the probability of nondegeneracy would be a welcome development beyond our present scope. 
    
    \subsection{Edge-Swap Markov Chains}

    	Edge-swap Markov Chain Monte Carlo provides an alternative approach to sampling from $\mu_{\mathbf{D}, \mathbf{K}}$. 
    	The benefit of this method of sampling is that, as long as it is initialized with a non-degenerate hypergraph, all samples produced are guaranteed to be nondegenerate. 
    	
        It is convenient to define edge swaps on the annotated bipartite graph $\B$. 
    	\begin{dfn}
    		An \emph{role-preserving edge-swap} on $\B$ is a map of pairs of edges: 
            \begin{align*}
                (v_1,e_1;x),(v_2,e_2;x) \mapsto  (v_2,e_1;x), (v_1,e_2;x) 
            \end{align*}
    	\end{dfn}
        We can also regard a role-preserving edge-swap of bipartite edges $f_1$ and $f_2$ as a map $\pi:\mathbb{B}\rightarrow \mathbb{B}$ that generates a new bipartite graph $\B'$. 
        In this case we write $\B' = \pi(\B|f_1, f_2)$.  
        Note that it is possible that $\B = \pi(\B|f_1, f_2)$; this occurs when $f_1 = (v, e_1;x)$ and $f_2 = (v, e_2;x)$ for some $v$ or $f_1 = (v_1, e;x)$ and $f_2 = (v_2, e;x)$ for some $e$. 
        Nondegeneracy rules out the case that $f_1 =  f_2 =  (v, e;x)$ for distinct $f_1$ and $f_2$. 

    	Let $\B^x$ denote the edges of $\B$ with role label $x$. 
    	Then, we can construct a Markov chain $\B_t\in \mathbb{B}$ by repeated role-preserving double edge swaps. 
    	Some care is needed to ensure that that each state of this chain is nondegenerate. 
    	The full algorithm is formalized in  \Cref{alg:MCMC}. 
		\begin{figure}
			\begin{algorithm2e}[H]
	           \DontPrintSemicolon
	            \caption{MCMC Sampling for $\mu$}\label{alg:MCMC}
	            \KwIn{degree sequence $\mathbf{d}$, initial annotated hypergraph $\h_0$ with bipartite graph $\B$,  sample interval $\delta t \in \mathbb{Z}_+$,  sample size $s \in \mathbb{Z}_+$.}
	            \textbf{Initialization:} $t \gets 0$, $\B \gets \B_0$\;
	            \While{$t \leq s(\delta t)$}{
	                sample $e_1,e_2$ u.a.r. from $\binom{\E_t}{2}$\; 
	                sample $f_1 = (v_1,e_1;x_1)$ and $f_2 = (v_2,e_2;x_2)$ u.a.r. from $e_1$ and $e_2$ \;
	                \uIf{$x_1 \neq x_2$}{pass}
	                \uElse{$\B' \gets \pi(\B_t|f_1,f_2)$\;
                        \If{$\B'$ is nondegenerate}{
                            $\B_{t+1} \gets \B'$\;
                            $t \gets t+1$
                        }
                    }
	            }
	            \KwOut{$\{\B_t \text{ such that } t|\delta t\}$}
	        \end{algorithm2e}
        \end{figure}
    	The Markov chain $\B_t$ of hypergraphs induces a chain $\h_t \in \mathbb{H}$ of annotated hypergraphs. 
    	\begin{thm}
    		The Markov chain $\h_t \in \mathbb{H}$ is irreducible and reversible with respect to $\mu$. 
    	\end{thm}
    	\begin{proof}
            It is convenient to first view stub-matching as an algorithm for generating stub-labeled bipartite graphs. 
            To do so, we observe that the output partition of stub-matching defines a bipartite graph similar to $\B$, except that to each edge $(v, e ;x)$ is associated an integer between $1$ and $d_v^x$. 
            Let $\bar{\B}$ denote this stub-labeled bipartite graph, and $\bar{\mathbb{B}}$ the set of such graphs. 
            We recover a standard bipartite graph $\B$ from $\bar{\B}$ by erasing the stub-labels. 

            We now observe that a bipartite edge-swap $\bar{\B}_{t+1} = \pi(\bar{\B}_t|f_1,f_2)$ always produces a new element of $\mathbb{B}$, since each bipartite edge has a distinct stub-label. 
            For the same reason, if $\bar{\B}_{t+1} = \pi(\bar{\B}_t|f_1,f_2)$, then $f_1$ and $f_2$ are the only bipartite edges for which this relation holds. 
            In particular, the transition kernel of the edge-swap Markov chain may be written
            \begin{align*}
                P(\bar{\B}'|\bar{\B}) = 
                    \begin{cases}
                        r(f_1, f_2|\bar{\B}) &\quad \exists f_1, f_2 : \bar{\B}' = \pi(\bar{\B}|f_1,f_2) \\
                        0 &\quad \text{otherwise,}
                    \end{cases}
            \end{align*}
            where $r(f_1, f_2|\bar{\B})$ is the probability that bipartite edges $f_1$ and $f_2$ are sampled and that $x_1 = x_2$. 
            It follows that $P$ will be reversible with respect to the  uniform measure on $\mathbb{B}$ provided that 
            $r(f_1,f_2|\bar{\B}) = r(f_1',f_2'|\bar{\B}')$ whenever $\bar{\B}' = \pi(\bar{\B}|f_1,f_2)$ and $\bar{\B} = \pi(\bar{\B}'|f_1',f_2')$.
            To see why this is the case, note that $r(f_1,f_2|\bar{\B})$ depends on $\bar{\B}$ only through the sizes of edges and the role distributions within each edge. 
            These quantities are preserved under double edge swaps. 
            We thus conclude that $\B_t$ is reversible with respect to $\mu_0$, and therefore $\h_t$ is reversible with respect to $\mu$.

            It remains to show irreducibility. 
            We will construct a supported path in state space from  $\bar{\B}$ to $\bar{\B}'$, where these are stub-labeled bipartite graphs with fixed marginals.  
            Let $\E$ and $\E'$ denote the respective edge sets of these bipartite graphs. 
            Choose $x$ such that $\bar{\B}^x \setminus \bar{\B}'^x$ is nonempty, and let $(v_1,e_1;x,i_1) \in \bar{\B}^x \setminus \bar{\B}'^x$. 
            Then, there must exist edges of the form $(v_1,e_2;x, i_1)$ and $(v_2,e_1;x, i_2)$ in $\bar{\B}'^x\setminus \bar{\B}^x$, since node $v_1$ must be connected to some edge with role $x$, and similarly edge $e_1$ must be connected to some node with role $x$. 
            In particular, the swap $(v_1,e_2; x, i_1), (v_2,e_1;x,i_2) \mapsto  (v_1,e_1;x,i_1), (v_2,e_2;x,i_2)$ reduces the size of the set $\bar{\B}^x \setminus \bar{\B}'^x$ by at least one. 
            Repeating this procedure allows us to reduce the size of $\bar{B}^x \setminus \bar{\B}'^x$ indefinitely, and therefore constitutes a supported path between them. 
    	\end{proof}

        \begin{cor}
            As $\delta t\rightarrow \infty$, the output of \Cref{alg:MCMC} is asymptotically independent and identically distributed according to $\mu$. 
        \end{cor}
        
        For the feature of guaranteed nondegeneracy, we pay a computational cost in mixing times. 
    	Unfortunately, no results are extant for this class of edge-swap Markov chains, while the best available upper bound for the mixing time of a related class of chains \cite{greenhill2014switch,greenhill2011polynomial} scales poorly with the node degrees and total number of edges. 
    	Despite this, edge-swap Markov chains can be deployed in a variety of practical settings, see \cite{Fosdick2018} for a review.

\section{Analysis of Annotated Hypergraphs}
\label{sec:observables}

    In this section, we introduce a series of tools for measuring the structural properties of annotated hypergraphs while flexibly incorporating information about roles. 
    We split these into two categories; those that can be natively defined on the annotated hypergraph structure, and those that can be measured using a projection of the annotated hypergraph to a weighted  directed network.

\subsection{Native Polyadic Observables}

\subsubsection{Role Densities}
    
    The simplest nontrivial statistic is the \emph{individual role density} associated to a node. 
    The individual role density is a probability distribution summarizing the proportion of interactions in which node $v$ plays each role. 
    It may be computed as 
    \begin{align}
        \label{eqn:individual_role_density}
        \mathbf{p}_v = \frac{\mathbf{d}_v}{\bracket{\mathbf{e}, \mathbf{d}_v}}\;,
    \end{align}
    where $\mathbf{e}$ is the vector of ones, and $\mathbf{d}_v$ is the vector of degrees of node $v$ in each role, defined in \Cref{eqn:node_degrees}. 

\subsubsection{Local Role Density}

    It is of interest to compare the individual role density $\mathbf{p}_v$ of node $v$ to the proportion of roles in a neighborhood of $v$. 
    Let $c^y_v$ give the number of times in which a node incident to $v$ plays role $y$, other than $v$ itself. 
    We then have, 
    \begin{align*}
         c^y_v = \sum_{x \in \X} \sum_{e \in \E} \mathbb{I}(\ell(v,e) = x)\sum_{u \in e, u \neq v}\mathbb{I}(\ell(u,e) = y) = \sum_{e: v \in e} k_e^y - d_v^y\;.
    \end{align*}
    Normalizing yields the \emph{local role density} $\mathbf{p}'$. 
    \begin{align}
        \label{eqn:local_role_density}
        \mathbf{p}'_v = \frac{\mathbf{c}_v}{\bracket{\mathbf{e}, c_v}}\;.
    \end{align}
    The local role density measures not the behavior of node $v$, but rather the typical behavior of the nodes with which $v$ interacts. 

\subsubsection{Assortativity}
    
    Classically, degree-assortativity measures the tendency of nodes of similar degrees to connect to each other. 
    In particular, it is often observed that nodes of high degree tend to connect to other nodes of high degree. 
    In an annotated hypergraph, each node possesses a distinct degree corresponding to each role. 
    We therefore develop a role-dependent assortativity measure, similar to assortativity for directed graphs. 
    The measure we choose generalizes that of \cite{chodrow2019configuration}, which was formulated for hypergraphs without annotations. 
    We first provide the mathematical formulation, and then discuss the nature of the correlation it measures. 

    Fix roles $x,y \in \X$. 
    For any nodes $u,v \in \V$, let 
    \begin{align*}
        s^{xy}_{uv} &= \sum_{e \in \E} \mathbbm{I}(\ell(u,e) = x)\mathbbm{I}(\ell(v,e) = y)
    \end{align*}
    count the number of edges in which  $u$ has role $x$ and $v$ has role $y$. 
    We compute an assortativity score as a correlation coefficient between the random variables $d_U^x - s^{xy}_{UV}$ and $d_V^y-s^{xy}_{UV}$, where the random nodes $U$ and $V$ are sampled according to a two-stage scheme.
    We first select a uniformly random edge $e$ from $E$. 
    From $e$, we then select a uniformly random pair of distinct nodes $U$ and $V$, conditioning the entire process on the roles $x$ and $y$. 
    The resulting probability law for $U$ and $V$ is 
    \begin{align}
        \prob^{xy}(U = u, V = v) = \frac{\sum_{e \in E} \binom{k_e}{2}^{-1} \mathbbm{I}(\ell(u,e) = x)\mathbbm{I}(\ell(v,e) = y)}{\sum_{e \in E} \binom{k_e}{2}^{-1} \mathbbm{I}(u \in e)\mathbbm{I}(v \in e)}\;. \label{eq:law}
    \end{align}
    To compute a Spearman assortativity coefficient, let $q_{uv}^x$ denote the the rank of $d_u^x - s^{xy}_{uv}$ among all pairs $u$ and $v$. 
    Then, the Spearman assortativity coefficient is given by  
    \begin{align}
        \rho^{xy} = \frac{\text{cov}\left(q_{UV}^x, q_{VU}^y\right)}{\sqrt{\var{q_{UV}^x}\var{q_{VU}^y}}} =  \frac{\mathbb{E}[(q_{UV}^x - \mathbb{E}[q_{UV}^x])(r_{VU}^y - \mathbb{E}[q_{VU}^x])]}{\sqrt{\mathbb{E}[(q_{UV}^x - \mathbb{E}[q_{UV}^x])^2]\mathbb{E}[(q_{VU}^y - \mathbb{E}[q_{VU}^y])^2]}}\;, \label{eq:spearman}
    \end{align}
    with expectations computed with respect to \Cref{eq:law}. 
    By construction, $\rho^{xy}$ is symmetric in the roles $x$ and $y$. 
    Note also that the definition of $q_{uv}^x$ excludes from the calculation instances in which $u$ and $v$ themselves interact in these roles.  
    The assortativity coefficient supports quantitative investigations of questions like: is it statistically the case that the sender and receiver of a given email tend to send and receive many emails, respectively? 
    Do scholars with many first-authorships tend to collaborate with scholars with many last-authorships? 

    Since it can in practice be difficult to compute the expectations appearing in \Cref{eq:spearman} exactly, it is often convenient to estimate them via repeated sampling from \Cref{eq:law}. 
    This is the method used in our experiments below. 

\subsection{The Weighted Projection} \label{sec:projection}

    We often wish to apply tools from dyadic graph theory and network science to polyadic data. 
    The usual way to do this is to \emph{project} the latter by replacing each $k$-dimensional hyperedge with a $k$-clique. 
    When role information is available, we can perform more flexible projections. 
    Let $\mathbf{R} = [r^{xy}] \in \R^{p\times p}$. 
    We refer to $\mathbf{R}$ as a \emph{role-interaction kernel}, which describes directed interaction strengths between pairs of roles. 
    We do not place any restriction on the values in $\mathbf{R}$, although we can without loss of generality rescale to ensure that $\max_{x,y}\abs{r^{xy}} = 1$. 
    In all our examples, the entries of $\mathbf{R}$ will be nonnegative, but in principle negative interaction weights can also be used. 

    We compute the weighted projection matrix $\mathbf{W} = \mathbf{W}(\h, \mathbf{R})$ entrywise via the formula
    \begin{align}
        w_{uv} = \sum_{x,y \in \X}r^{xy} \sum_{e \in \E} \mathbb{I}(\ell(u,e) = x)\mathbb{I}(\ell(v,e) = y)\;. \label{eqn:weighted_network}
    \end{align}
    Each entry $w_{uv}$ thus counts the number of edges connecting any pair of nodes, weighted by the entries of $\mathbf{R}$. 
    To illustrate, let us first consider the directed hypergraphs of \cite{gallo1993directed}. 
    
    Directed hypergraphs possess two possible roles, ``source'' and ``target.''
    To project a directed hypergraph to a weighted dyadic graph that respects the roles, we can compute  \Cref{eqn:weighted_network} using the interaction kernel 
    \begin{align*}
        \mathbf{R} = \kbordermatrix{
            & \text{source} & \text{target}  \\ 
            \text{source}   & 0 & 1  \\ 
            \text{target}   & 0 & 0
        }\;. 
    \end{align*}
    The result is a weighted directed graph in which $w_{uv}$ counts the number of hyperedges in which $u$ appears as a source and $v$ as a target. 
    This example is somewhat trivial, but the benefit of our general formalism is that flexible modeling choices are possible. 
    For example, in our case study of the Enron email data set below we use the interaction kernel 
    \begin{align*}
        \mathbf{R} = \kbordermatrix{
            & \text{from} & \text{to} & \text{cc} \\ 
            \text{from} &0 & 1 & 0.25\\ 
            \text{to} &0 & 0 & 0\\ 
            \text{cc} &0 & 0 & 0
        }\;. 
    \end{align*}
    This kernel reflects an assumption that information travels efficiently from senders to direct receivers, but more weakly to cc'd receivers. 

\subsubsection{Centrality}
    
    Standard, dyadic centrality analysis may be performed on the weighted projected graph of an annotated hypergraph. 
    For our examples in \Cref{sec:results} we consider the results of computing Pagerank \cite{page1999pagerank} and eigenvector centrality on weighted projected networks \cite{newman2010networks}.
    We note that there are alternative approaches to defining centrality in hypergraphs such as \cite{zhou2007learning} and \cite{benson2019three}. 
    The first of these uses a random walk formulation that can be represented via normalized weighted projections. 
    The latter is applicable only for $k$-uniform hypergraphs. 
    Our approach, while not a direct generalization of either, is considerably more flexible in applied data analysis.

\subsubsection{Modularity and Community Detection}

    We extend the notion of modularity to annotated hypergraphs. 
    Many such extensions are possible. 
    Unlike a recent proposal \cite{Kami2018}, we define a dyadic notion of modularity via null expectations for the weighted projected graph computed via \Cref{eqn:weighted_network}. 
    While our approach loses some higher-order information in the modularity calculation, it has the benefit of allowing roles to be flexibly incorporated into the null expectation.

    Recall that $w_{uv}$ gives the observed weighted edge count from $u$ to $v$, with the weights specified via $\mathbf{R}$. 
    In order to derive a working notion of dyadic modularity, we need only to estimate the expectation of $w_{uv}$ under a suitably-chosen null. 
    We will approximate this expectation under the annotated hypergraph configuration model. 
    
    Let us estimate $\E_\mu[M_{uv}^{xy}]$, the expected number of edges that contain node $u$ in role $x$ and node $v$ in role $y$. 
    We begin by forming an edge $e$ via stub-matching. 
    There are $k^x_e$ $x$-stubs that must be selected to form $e$. 
    Each of these has probability approximately $\frac{d^x_u}{\sum_{\ell}d^x_\ell} = \frac{d^x_u}{\bracket{\mathbf{e}, \mathbf{d}^x}}$ to be node $u$.
    Supposing the probability of degeneracy in an individual edge to be small, we can thus approximate the probability that $e$ contains $u$ in role $x$ as $k^x_e\frac{d^x_u}{\bracket{\mathbf{e}, \mathbf{d}^x}}$. 
    Similarly, the probability that $e$ contains $v$ in role $y$ is $k^y_e\frac{d^y_v}{\bracket{\mathbf{e}, \mathbf{d}^y}}$.
    Summing over edges gives our approximation for $\E_\mu[M_{uv}^{xy}]$: 
    \begin{align*}
        \E_\mu[M_{uv}^{xy}] &= \sum_{e \in E} k_e^xk_e^y \frac{d_u^xd_v^y}{\bracket{\mathbf{e}, \mathbf{d}^x}\bracket{\mathbf{e}, \mathbf{d}^y}}\;.
    \end{align*}
    We can write this expression more compactly as 
    \begin{align*}
        \E_\mu[\mathbf{M}^{xy}] = \frac{\bracket{\mathbf{k}^x, \mathbf{k}^y}(\mathbf{d}^x\otimes \mathbf{d}^y)}{\bracket{\mathbf{e}, \mathbf{d}^x}\bracket{\mathbf{e}, \mathbf{d}^y}} \;,
    \end{align*}
    where $\otimes$ denotes the vector outer product.

    To compute $\E_\mu[\mathbf{W}]$, we weight by $\mathbf{R}$ and sum over role pairs: 
    \begin{align*}
        \E_\mu[\mathbf{W}] = \sum_{x,y\in \X} r^{xy} \E_\mu[\mathbf{M}^{xy}]\;.
    \end{align*}
    We then define a dyadic modularity score of a partition $g$:  
    \begin{align}
        Q_\mu(g) = \frac{1}{\bracket{\mathbf{e},\mathbf{W}\mathbf{e}}} \trace{\mathbf{G}^T(\mathbf{W} - \mathbb{E}_\mu[\mathbf{W}]) \mathbf{G}}\;, \label{eq:modularity}
    \end{align}
    where $\mathbf{G}$ is the one-hot encoding matrix of the partition $g$. 
    As usual, the pre-factor ensures that $-1 \leq Q_\mu(g) \leq 1$. 
    
    It is important to clarify the nature of the null model used in the modularity calculation. 
    A procedure that is commonly followed for studying polyadic data is to construct a projected graph and perform modularity maximization with respect to an implicit null defined over dyadic graphs. 
    In contrast, we have defined a null over the space of annotated hypergraphs, and then computed expectations in the projected graph with respect to this higher-order null. 
    This approach has the benefit of preserving some information about polyadic interactions in the modularity score, even though this score is natively dyadic. 
    These two approaches will generally lead to different null matrices and therefore different partitions. 
    
    In order to approximately maximize \Cref{eq:modularity}, we adopt the multiway spectral algorithm of \cite{zhang2015multiway}, though many alternatives are possible. 
    The matrix $\mathbf{W}$ is not symmetric, and therefore it is necessary to make a small adjustment to this algorithm \cite{leicht2008community}. 
    Rather than computing the leading eigenvectors of the matrix $\mathbf{W} - \mathbb{E}_\mu[\mathbf{W}]$, we instead compute the eigenvectors of the symmetrized form $\frac{1}{2}(\mathbf{W} + \mathbf{W}^T - \mathbb{E}_\mu[\mathbf{W} + \mathbf{W}^T])$. 
    The multiway spectral algorithm is then applied to perform this task. 
    
\section{Results}
\label{sec:results}

    We illustrate how annotated hypergraphs and their null models can be used to enrich analysis of previously studied data sets.
    
\subsection{Enron Case Study}

    We first focus our analysis on the Enron email data set \cite{klimt2004introducing}.
    This data contains the emails from employees of the company prior to its forced bankruptcy and shutdown due to corporate fraud and corruption.
    While this data is temporal in nature we consider only the time-aggregated graph, neglecting the timestamps of edges.
    We also only consider official email accounts of Enron employees, referred to as the `core' group.

\subsubsection{Role Distributions and Assortativity}

    \begin{figure}
        \centering
        \begin{subfigure}{0.32\linewidth}
        \centering
        \includegraphics[width=\linewidth]{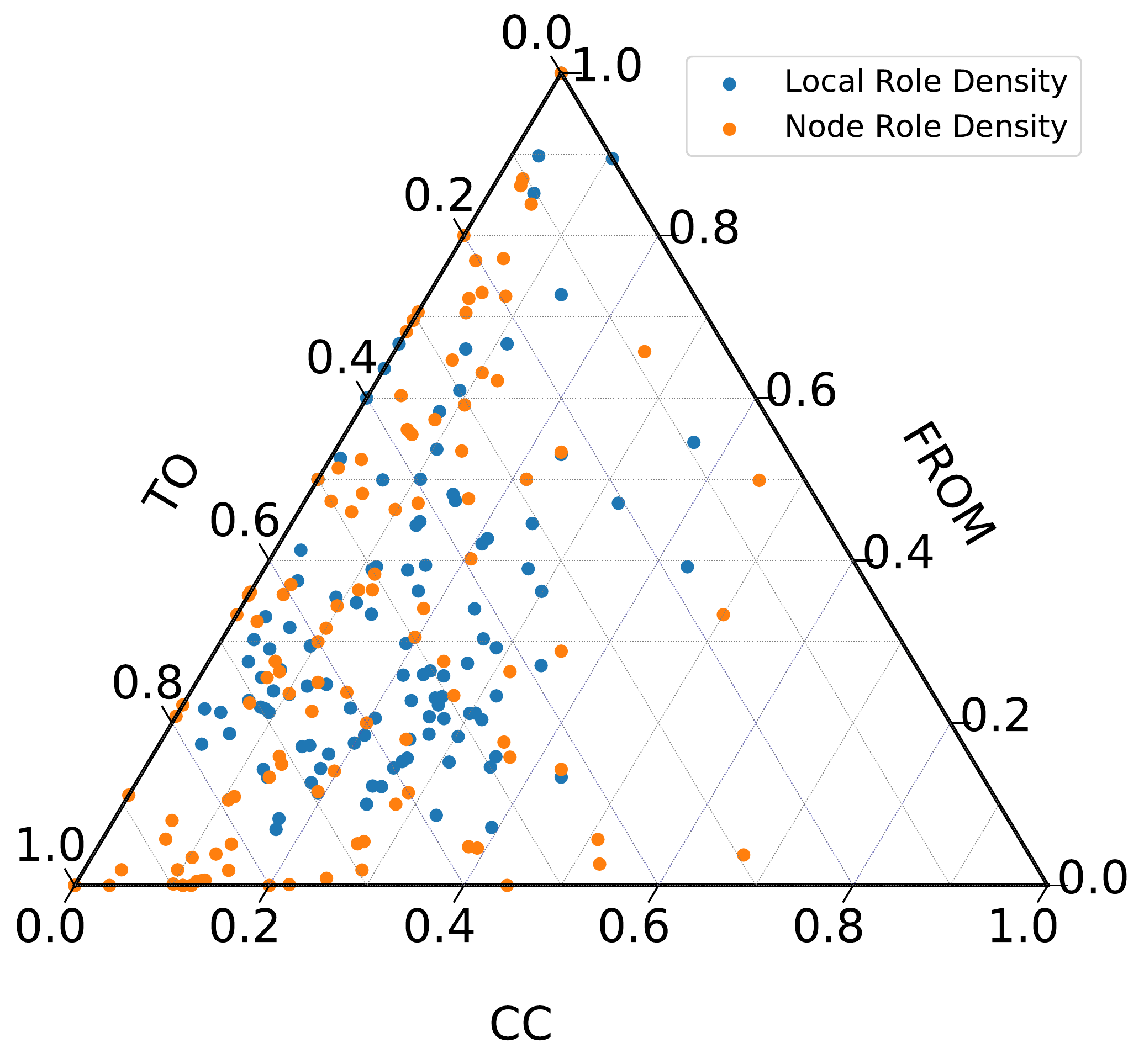}        
        \caption{}
        \end{subfigure}
        \begin{subfigure}{0.32\linewidth}
        \centering
        \includegraphics[width=\linewidth]{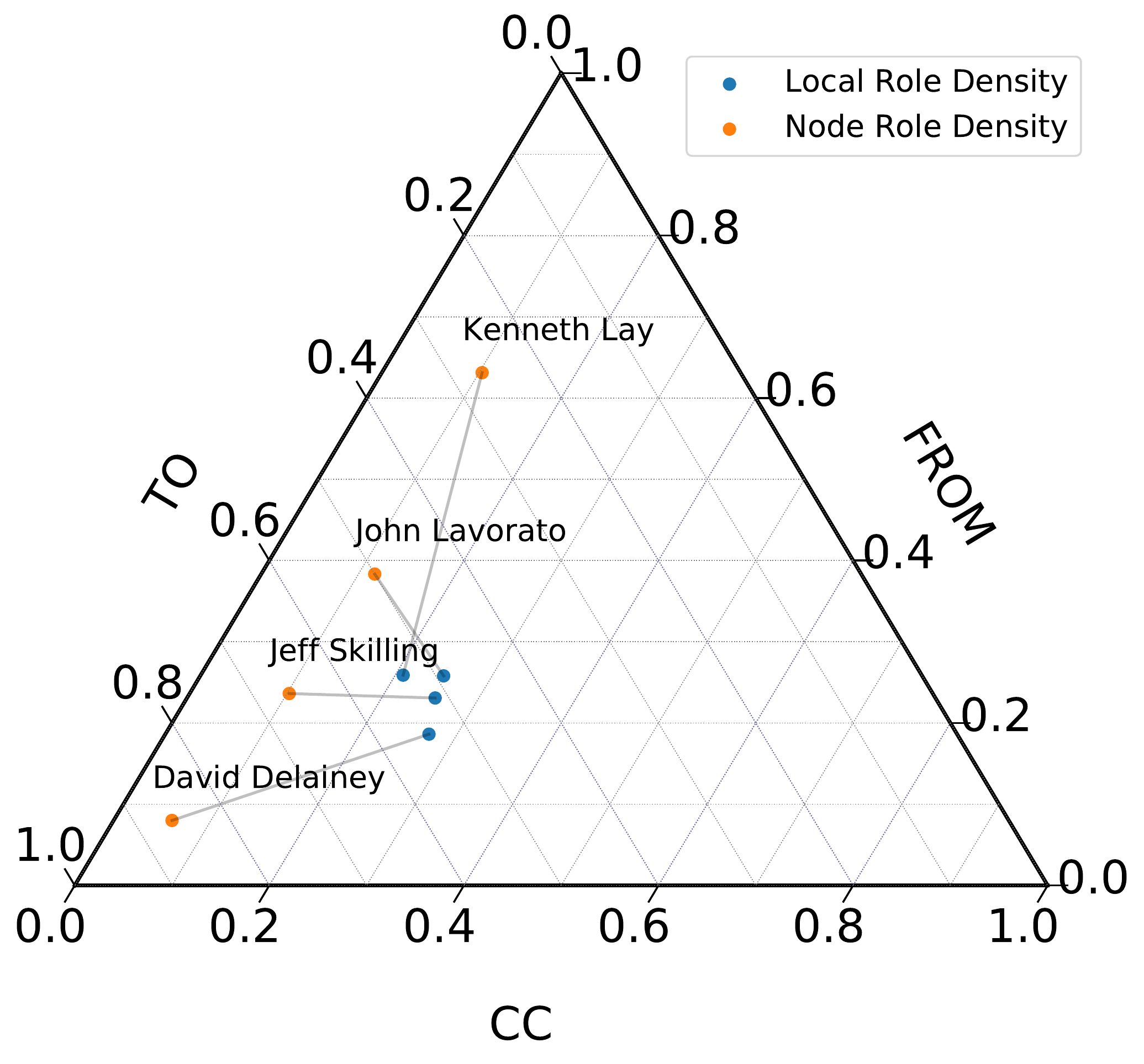}        
        \caption{}
        \end{subfigure}
        \begin{subfigure}{0.32\linewidth}
        \centering
        \includegraphics[width=\linewidth]{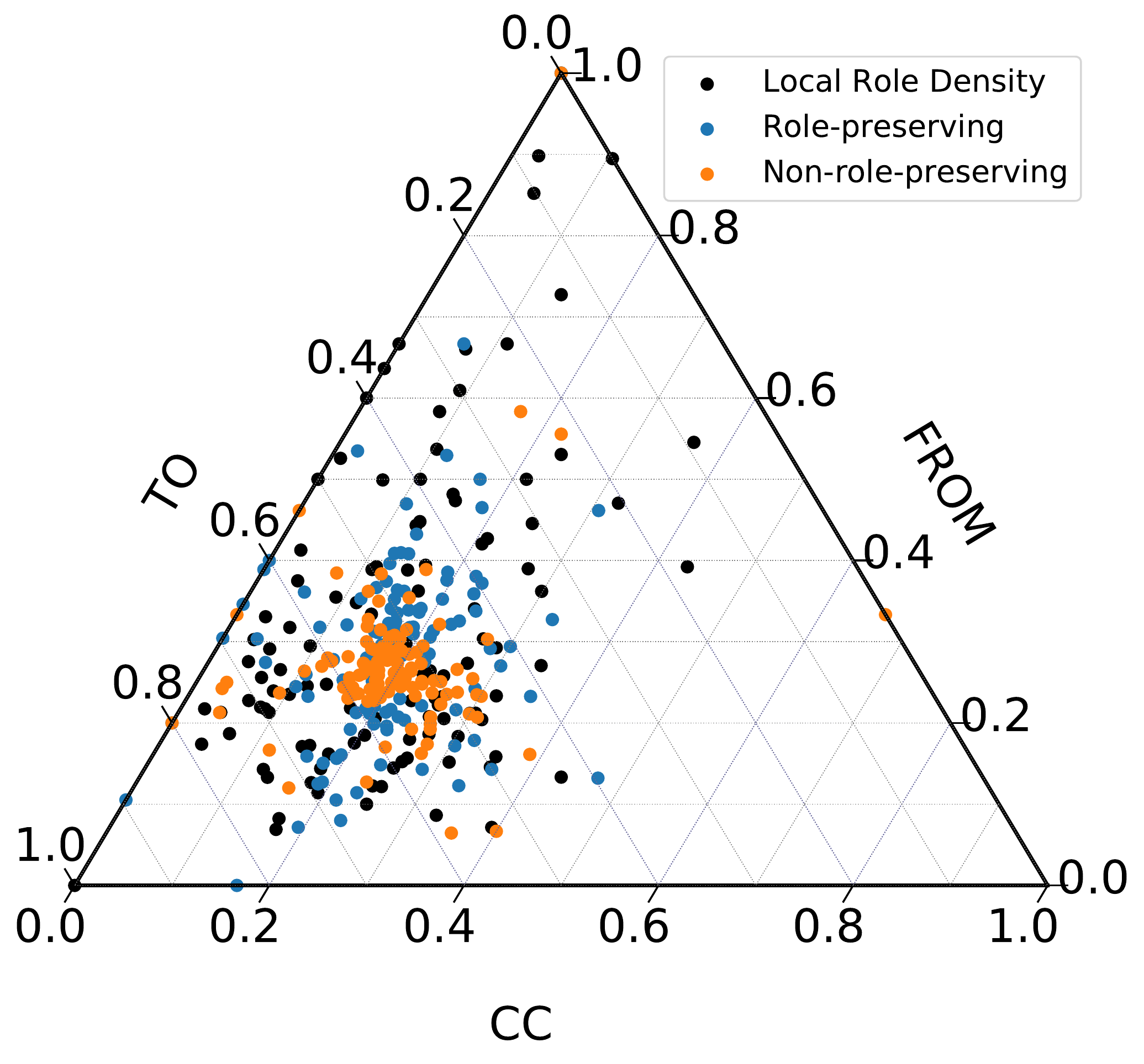}        
        \caption{}
        \end{subfigure}
        \caption{
            Role distributions in the Enron data set.
            (A). The individual role densities (orange) and local role densities (blue).
            (B). The individual role densities (orange) and local role densities (blue) of the four Enron CEOs, now with lines connecting the statistics for each node.
            (C). The local role densities (black) compared with a single sample from the role-preserving ensemble (blue) and non-role-preserving ensemble (orange). 
        }
        \label{fig:enron_roles}
    \end{figure}

    The data contains three roles: ``from'', ``to'', and ``cc.''
    These describe the sender, recipients, and carbon copy recipients of emails respectively\footnote{
        Blind carbon copy, or ``bcc,'' has been merged into `cc' for simplicity.
    }.
    Depending on the occupation of the employee, the number of emails they send and receive (and therefore their role participation) will vary.
    \Cref{fig:enron_roles}(A) shows the distribution individual and local role densities for the data.
    Here we see a diverse range of behaviour, with an increased diversity in individual node distributions compared to local distributions.
    For example, there are nodes that have exclusively message senders and exclusive receivers, but no node neighbourhoods consist of exclusively one role.
    The extent of role hybridization highlights the importance of modeling roles as properties of node-edge pairs, since no node can be assigned a single role. 

    \Cref{fig:enron_roles}(B) shows the individual and local role densities for the four CEOs of the company, now each connected by a line.
    Here we see that, while the CEOs correspond with other individuals who perform similar roles, they exist across a spectrum of behaviour themselves, namely in whether they are senders or receivers of emails.
    For example, Kenneth Lay was the sender of emails roughly $65\%$ of the time, in comparison with David Delainey, who played that role in less than $10\%$ of interactions.
    This passive role of David Delainey as a receiver of information rather than originator of it may be reflected in the lesser sentencing he received in comparison to Lay. 
    
    In \Cref{fig:enron_roles}(C) we show the effect of randomization under hypergraph configuration models. 
    The configuration model introduced in \Cref{sec:definitions} preserves the roles of all nodes, while in the non-role-preserving variant we simply erase the role labels. 
    The samples from both null models show a reduction in heterogeneity of the role distribution -- randomization has the effect of homogenizing the population. 
    This effect is especially apparent under non-role-preserving randomization and the local role distribution converges towards the average over all nodes. 
    
    \begin{figure}
        \centering
        \includegraphics[width=\linewidth]{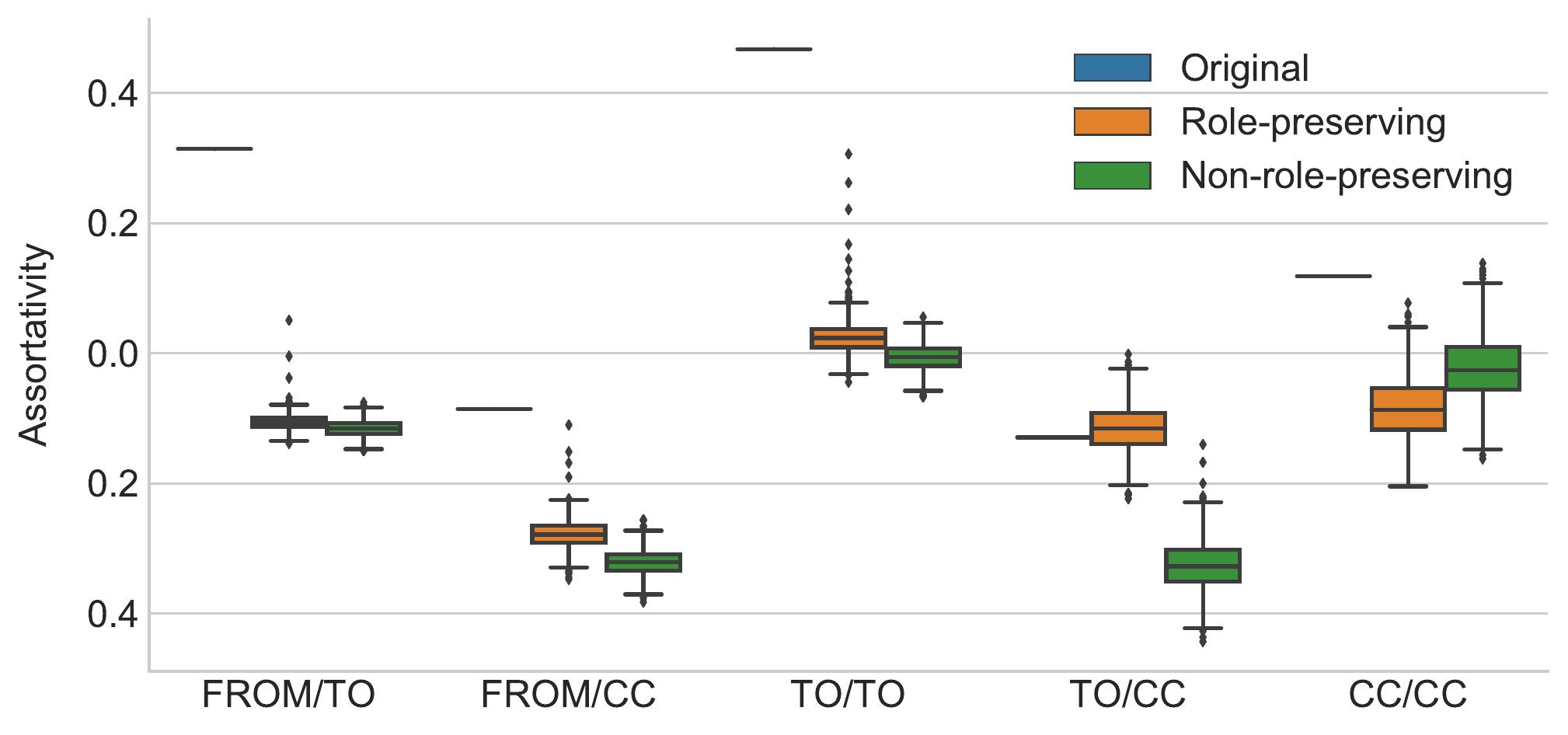}
        \caption{
            Hypothesis-testing for significant role-dependent assortativity as measured by \Cref{eq:spearman}.  
            The empirical values (black bars) are shown alongside labeled (orange) and unlabeled (green) configuration null distributions for each combination of roles.
        }
        \label{fig:enron_assortativity}
    \end{figure}

    The role assortativity, defined in \Cref{eq:spearman}, captures the tendency of nodes that frequently play one role to be associated with nodes that frequently play another. 
    \Cref{fig:enron_assortativity} compares the the observed and null-distributed assortativities $\rho^{xy}$ for each combination of roles in the Enron data set. 
    Five combinations are shown -- there are six pairs of role labels, and the `from/from' combination is vacuous since all emails have exactly one sender. 
    In four of the five combinations, the assortativity coefficient is much higher than would be expected under the null and would generally be judged statistically significant. 
    The positive assortativities in the `from/to' and `from/cc' combinations quantify the tendency of prolific senders to share information with prolific receivers.
    The latter combination highlights the importance of using null models to contextualize network measurements. 
    Despite the fact that the observed assortativity is negative, the null distributions are even more so. 
    The data should therefore be judged more assortative than expected by chance. 
    The `to/to' and `cc/cc' combinations quantify the tendency of important receivers to do so along the same communication threads.
    The `to/cc' combination illustrates the utility of role-preserving randomization -- while this measurement would be statistically significant under randomization without role information, it falls within the bulk of the null when roles are preserved. 
    We recall that the Spearman coefficient defined above subtracts out the edges along which $u$ and $v$ interact. 
    It is natural to conjecture that the result for the `to/cc' combination indicates a tendency for nodes to cluster in recurring `to/cc' motifs, which are then removed in the course of calculation. 
    An example of a recurring pattern might be frequent emails to an executive with their personal assistant cc'd. 
    
\subsubsection{Centralities}
    In this and subsequent sections, we study the weighted projected graph of the Enron annotated hypergraph. 
    We will primarily use the role interaction kernel given by 
    \begin{align}
        \mathbf{R} = \kbordermatrix{
            & \text{from} & \text{to} & \text{cc} \\ 
            \text{from} &0 & 1 & 0.25\\ 
            \text{to} &0 & 0 & 0\\ 
            \text{cc} &0 & 0 & 0
        }\;. \label{eq:info_kernel}
    \end{align}
    This kernel emphasizes flow along edges -- information flows strongly to receivers listed in the `to' field and less strongly to those listed in the `cc' field. 
    
    \Cref{fig:enron_centralities} illustrates the flexibility of interaction kernels in studying graph properties. 
    We compute eigenvector and PageRank centralities on the weighted projected graph using the kernels $\mathbf{R}$ and $\mathbf{R}^T$. 
    High-centrality nodes under  $\mathbf{R}$ will tend to be those to whom information flows, while under $\mathbf{R}^T$ they will tend to be those from whom information originates. 
    The kernel $\mathbf{R}$ thus emphasizes information \emph{sinks}, and $\mathbf{R}^T$ information \emph{sources}. In bulk, the source and sink centralities are only weakly correlated, indicating that they capture distinct structural properties of the network. 
    The flexibility of the formalism of annotated hypergraphs with user-specified role interaction kernels supports the discovery of these features. 
    
    \begin{figure}
        \centering
        \includegraphics[width=.7\linewidth]{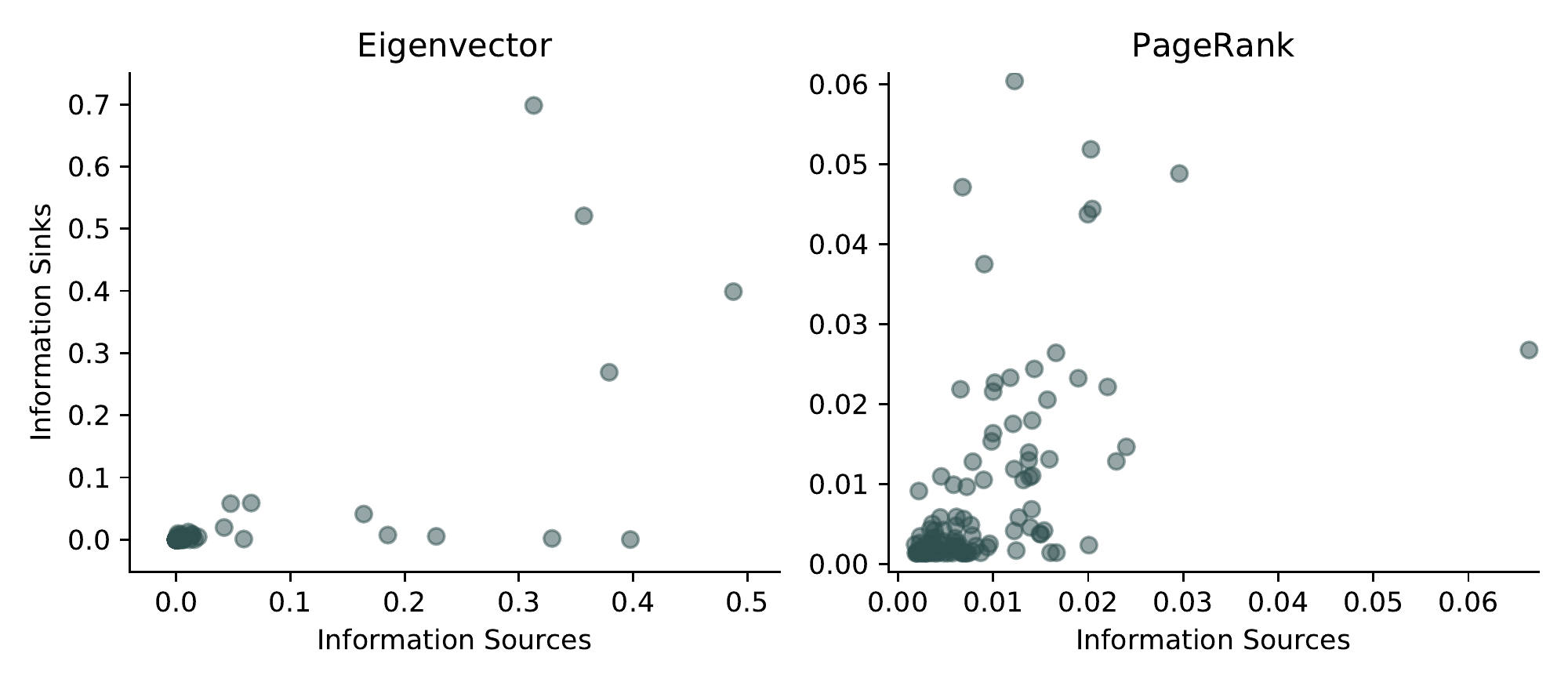}
        \caption{
            Sensitivity of eigenvector and PageRank centrality measures to role interaction kernels. 
            The horizontal axis gives centrality scores under the projection $\mathbf{R}^T$, while the vertical gives those under $\mathbf{R}$. 
            The PageRank teleportation parameter is $\alpha = 0.15$. 
        }\label{fig:enron_centralities}
    \end{figure}

\subsubsection{Modularity Maximization}

    \begin{figure}
        \centering
        \includegraphics[width=.49\textwidth]{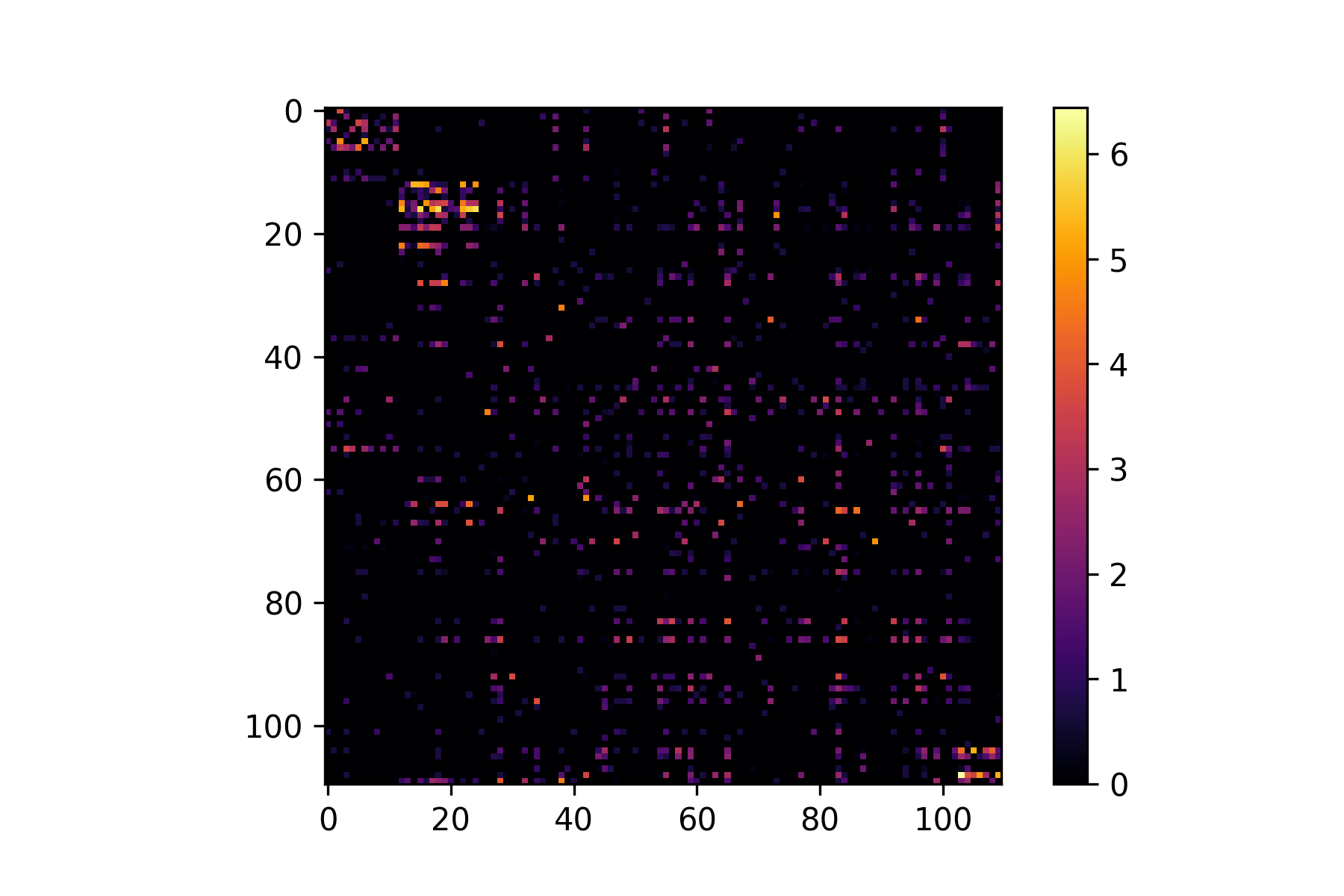}
        \includegraphics[width=.48\textwidth]{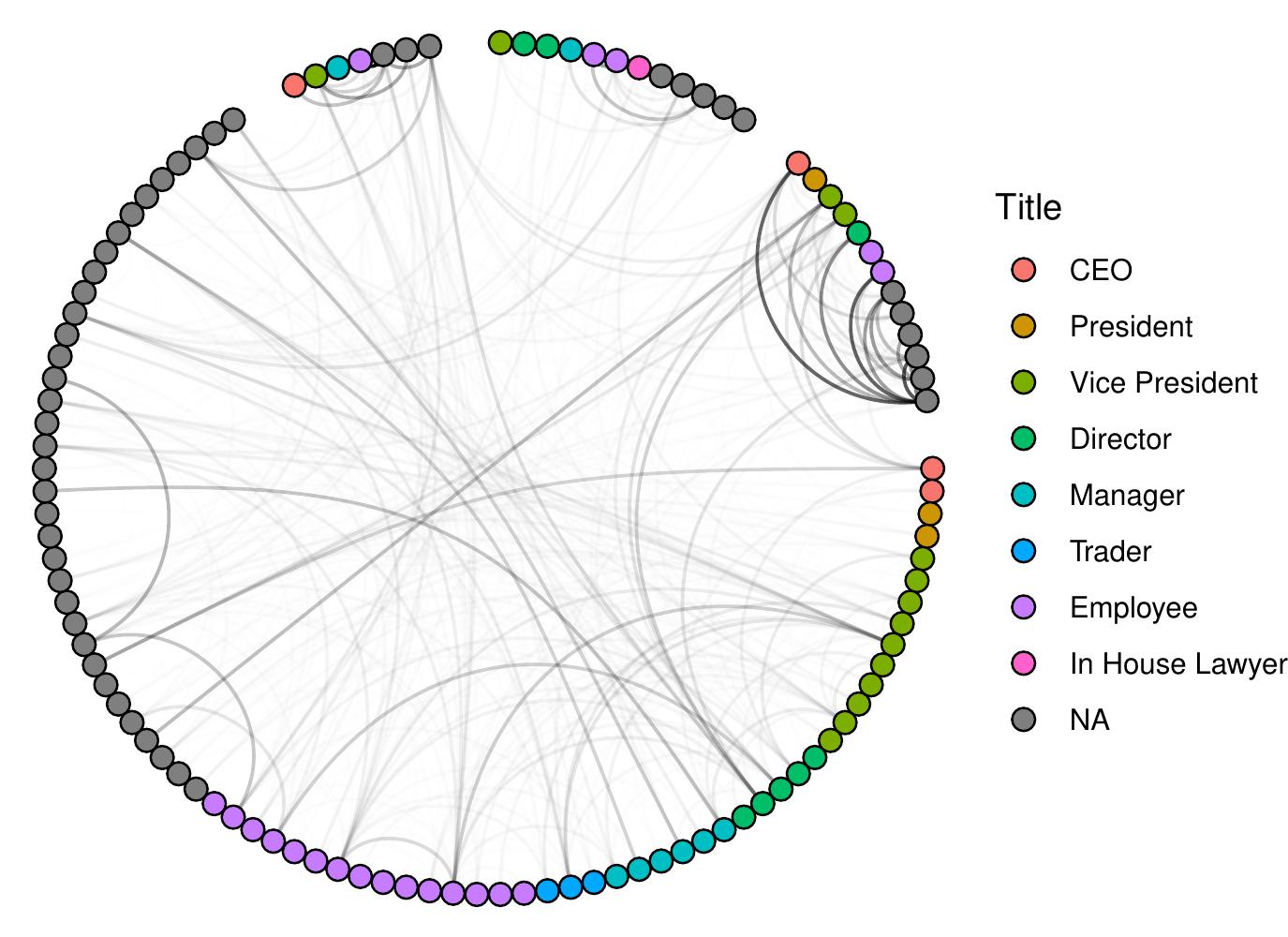}
        \caption{
        Example partition of the Enron data set via approximate spectral maximization of \Cref{eq:modularity}. 
        (Left): The adjacency matrix, arranged in order of the partition. 
        Colors are shown on a log-scale. 
        (Right): Visualization of the network. 
        Directed edges have been made undirected. 
        The darkness of each edge corresponds to its weight. 
        The color of each node corresponds to listed job title. 
        In this experiment, $k = 4$, and the kernel used is the information flow kernel $R$ of \Cref{eq:info_kernel}. 
        The modularity of the shown partition is 
        $0.506$. 
        }
        \label{fig:enron_communities}
    \end{figure}  
    
    \Cref{fig:enron_communities} shows the best of 100 partitions obtained using this kernel for $k = 4$ communities. 
    Inspection of the clustered adjacency matrix (left) suggests that three of the communities are relatively coherent, while one is highly disperse and essentially serves as a ``none of the above'' class. 
    On the right, we show the communities themselves, along with nodes colored according to job title. 

    It is useful to draw contrast against the uniform role interaction kernel, which ignores roles entirely: 
    \begin{align}
        \mathbf{R}' = \kbordermatrix{
            & \text{from} & \text{to} & \text{cc} \\ \text{from} &1 & 1 & 1\\ 
            \text{to} &1 & 1 & 1\\ 
            \text{cc} &1 & 1 & 1
        } \label{eq:uniform}
    \end{align}
    The use of this kernel produces substantially different communities when compared to the information flow kernel. 
    The best of 100 partitions obtained using this kernel only chose three communities, when up to $k = 4$ were available. 
    This partition achieves modularity $Q = 0.524$.
    The modularity scores of the two kernel weighting methods should not be compared directly, since they are computed on different modularity matrices. 

    A simple measure of similarity between the two partitions is the normalized mutual information 
    \begin{align*}
        NMI = 2\frac{I(X,Y)}{H(X) + H(Y)}\;,
    \end{align*}
    where $X$ and $Y$ are random variables giving the community assignment of a uniformly random node under each scheme, $H(X)$ is the entropy of random variable $X$, and $I(X,Y)$ the  mutual information of the two random variables $X$ and $Y$. 
    The NMI has maximum value of unity when $X$ and $Y$ are deterministically related, and minimum value of zero when they are statistically independent. 
    In this case, the normalized mutual information between the weighted and unweighted community assignments is 0.55. 
    This score indicates that the partitions are correlated, as we might expect -- however, there are nevertheless substantial differences between them. 
    These simple experiments emphasize the importance of appropriately-specified interaction kernels when performing community detection on annotated hypergraphs. 

\subsection{Ensemble Study}

    We now turn to studying annotated graph features systematically across a range of data sets.
    We consider six data sets, outlined in \Cref{tab:dataset_stats}, with full descriptions given in \Cref{app:data}. 
    We capture the wide variety of possible hypergraphs, those where the number of nodes exceeds that of the number of edges ($|V| \gg |E|$), and those where the number of nodes is much less than the number of edges ($|V| \ll |E|$).
    The number of node-edge stubs can take a maximum value of $|V||E|$ which corresponds to every node being included in every edge. 
    All the data are social in nature, but diverse in their interpretation. 
    With the exception of the Enron data, all are sparse, with low mean node degree $\bracket{d}$ and edge dimension $\bracket{k}$. 

    \begin{table}[tb]
        \setlength{\tabcolsep}{5pt}
        \centering
        \scalebox{1.0}{
        \begin{tabular}{l r r r r r }
        \toprule
                       & $\abs{\V}$     & $\abs{\E}$    & $\abs{\X}$    & $\bracket{d}$       & $\bracket{k}$  \ \\
        \midrule
        \texttt{enron}          & 112 & 10,504  & 3  & 230.6   & 2.5  \\
        \texttt{stack-overflow}     & 22,131 & 4,716  & 3 & 1.3   & 6.0  \\
        \texttt{math-overflow}     & 410 & 154 & 3 &  1.7 & 4.6  \\
        \texttt{scopus-multilayer}       & 1,677   & 938 & 3 & 1.7 & 3.1  \\
        \texttt{movielens}    & 73,155    & 43,058 & 2 & 2.8 & 4.7  \\ 
        \texttt{twitter}           & 52,294 & 123,158 & 4 & 5.1 & 2.2   \\
        \bottomrule
        \end{tabular}
        }
        \caption{Descriptive statistics for each data set. The fourth and fifth columns give the mean node degree and mean edge dimension, both ignoring role labels.}
        \label{tab:dataset_stats}
    \end{table}

    For each data set we calculate seven summary statistics.
    Two of these are the average entropy of the individual and local role densities (\Cref{eqn:individual_role_density,eqn:local_role_density} resp.), which capture the diversity of roles observed on individual nodes and their neighbourhoods.
    An entropy of zero indicates no role diversity observed, while maximal entropy\footnote{
        If the number of roles is $|\X|=p$ then the entropy has a maximum value of $\log_2 p$.
    }
    indicates all roles are equally likely to be observed.
    In addition we calculate the mutual information between a node's individual roles and local roles.
    This quantity is computed as the mean KL-divergence between the individual and local role densities. 
    We also calculate the weighted degree, PageRank, and eigenvector centralities for each node.
    To create suitable summary statistics we again use entropy, now across the distribution of centrality across nodes.
    This captured how concentrated the centrality is across all nodes in the hypergraph.
    Finally we report the number of weakly connected components in the weighted projection.

    We assess the significance of each feature by comparing with ensembles generated from both role-preserving and non-role-preserving randomised swaps.
    We take $500$ samples from each null model, each time performing $\lfloor 0.1|\mathcal{E}'| \rfloor$ shuffles, where $|\mathcal{E}'|$ is the number of edge stubs.
    We use a burn-in period of $10|\mathcal{E}'|$ shuffles to ensure that all chains are sufficiently well-mixed. 

    \begin{figure}
        \centering
        \includegraphics[width=\linewidth]{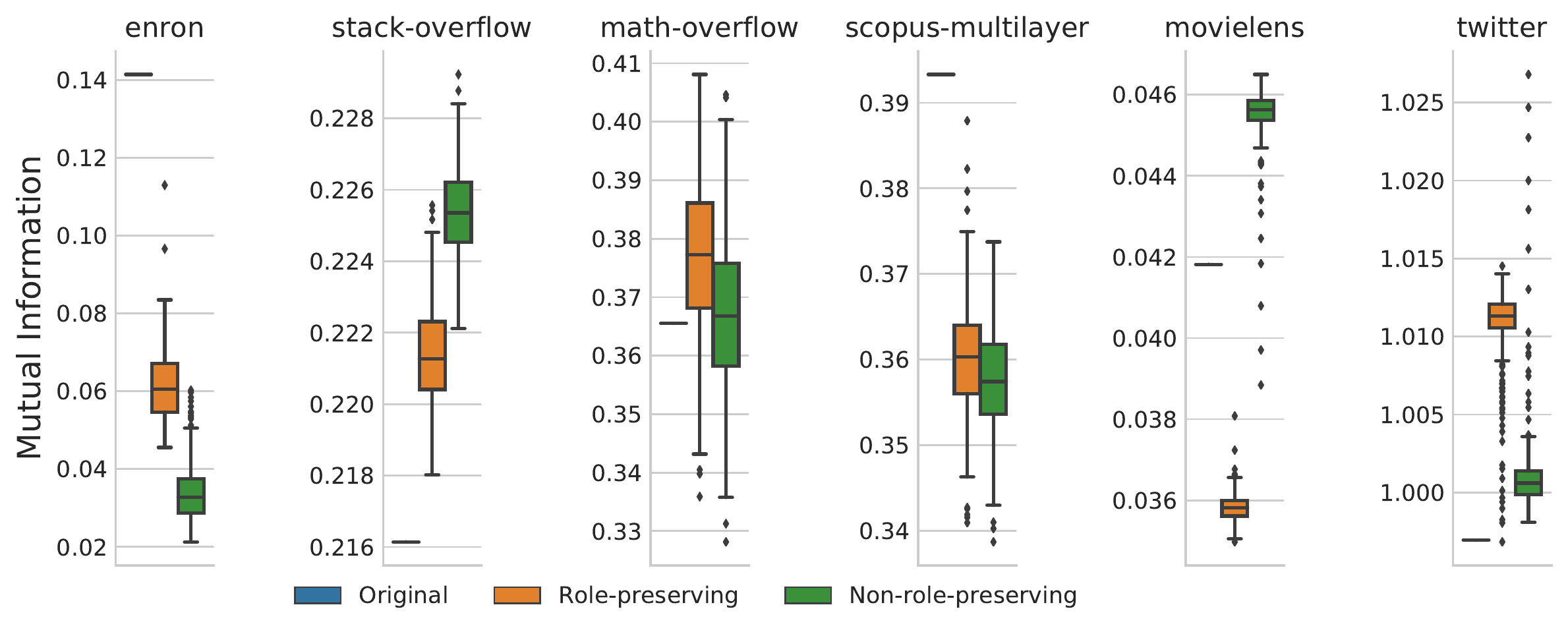}
        \caption{The average local role mutual information across data sets, with comparisons to the role-preserving (green) and non-role-preserving ensembles (green).
        In all cases, the observed value is falls outside the inter-quartile range for the role-preserving ensemble (orange), and in all but the \texttt{math-overflow} data set for the non-role-preserving ensemble (green).
        Note that in both the \texttt{movielens} and \texttt{twitter} data sets we may not have completely reached the mixed-state during the burn-in period, however in both cases the ensemble distributions differ significantly from the empirical observation.
        }
        \label{fig:ensemble_distributions_examples}
    \end{figure}

    In \Cref{fig:ensemble_distributions_examples} we show the local role mutual information across each data set.
    We see in all cases except for \texttt{math-overflow} that the local role mutual information is significant when compared to the non-role-preserving null model.
    This intuitively makes sense given that nodes may switch roles and so any correlation between node states is lost.
    In all data the local role density is significant when compared to the role-preserving model, however again the \texttt{math-overflow} shows the least difference (likely due to the small data size).
    Despite being significant in all but one data set, the local role mutual information can be both larger than expected (e.g. \texttt{enron, scopus-multilayer}) and small than expected (e.g. \texttt{stack-overflow, twitter}) when compared to the null.
    This can be explained by certain nodes having little diversity of roles in their local neighbourhood.
    For example, in the \texttt{enron} hypergraph, certain nodes may only ever send messages (and so their neighbourhood is considers of recipients) as we saw in \Cref{fig:enron_roles}(b).
    However, upon shuffling, these nodes may be swapped with other nodes with more diverse neighbourhoods, effectively increasing the average local diversity.
    Those which are lower than expected suggest that the hyperedges themselves are relatively uniform in participating roles. 
    In the \texttt{twitter} data this could be due to the limit on the edge cardinality (due to the character limit on posts).

    The significance of the remaining features for two data sets is given in \Cref{tab:ensemble_significance}.
    Naturally the node role entropy is never significant under the role-preserving null model, however in all cases it becomes significant when roles are not preserved upon shuffling.
    In the \texttt{enron} data set all features (barring the number of connected components) differ significantly from the non-role-preserving null.
    Interestingly the shuffling effect is not consistent across the different centrality measures.
    For the eigenvector centrality is more localised than expected (lower entropy) however the PageRank is less localised than expected (higher entropy). 
    Echoing \Cref{fig:enron_roles}(C), the neighbourhood role distributions are more diverse in both null models. 
    For the \texttt{stack-overflow} data, the significance of number of components can easily be explained by the presence of topic cliques in the original data.
    For example, users answering questions on \emph{Python} may be unlikely to answer questions on \emph{Javascript}.
    Since the null model is agnostic to these topics, the components are merged under shuffling.
    The role entropies are all significantly lower than expectation.
    This can be explained by nodes being consistent in the the roles that they take across multiple questions - that is, question answerers tend to answer more questions, for example.

    \begin{table}[]
        \setlength{\tabcolsep}{5pt}
        \centering
        \caption{
            The significance of multiple observables across the \texttt{enron} and \texttt{stack-overflow} data sets.
            Significance scores coloured red are two standard deviations larger in the original data than the null model.
            In contrast, those coloured blue are two standard deviations smaller.
            The results for all data sets are presented in \Cref{fig:ensemble_full}.
            }
        \scalebox{0.81}{
\begin{tabular}{llcccccc}
	                                         &                             &           & \multicolumn{2}{c}{Non-preserving}                             & \hspace{5pt}                       & \multicolumn{2}{c}{Preserving}                                        \\
	\cline{4-5}\cline{7-8}
	Data                                     & Feature                     & Orig. Val & Avg. Val.                      & $z$                       &                                    & Avg. Val. & $z$                      \\
	\midrule
	\multirow{7}{*}{\texttt{enron}}          & Connected Components        & 1.00      & 1.00                           & 0.00                      &                                    & 1.00      & 0.00                     \\
	                                         & Local Role MI               & 0.14      & 0.03                           & \textcolor{red}{15.68}    &                                    & 0.06      & \textcolor{red}{9.78}    \\
	                                         & Local Role Entropy         & 1.24      & 1.36                           & \textcolor{blue}{-12.47}  &                                    & 1.34      & \textcolor{blue}{-8.24}  \\
	                                         & Node Role Entropy           & 0.99      & 1.33                           & \textcolor{blue}{-19.55}  &                                    & 0.99      & 0.00                     \\
	                                         & Weight. Degree Entropy      & 5.87      & 5.83                           & \textcolor{red}{3.00}     &                                    & 5.87      & -0.39                    \\
	                                         & Weight. Eigenvector Entropy & 4.02      & 5.84                           & \textcolor{blue}{-154.46} &                                    & 5.88      & \textcolor{blue}{-55.09} \\
	                                         & Weight. Pagerank Entropy    & 6.36      & 6.11                           & \textcolor{red}{26.47}    &                                    & 6.14      & \textcolor{red}{13.74}   \\
	\midrule
	\multirow{7}{*}{\texttt{stack-overflow}} & Connected Components        & 1168.00   & 1166.71                        & 0.06                      &                                    & 1041.01   & \textcolor{red}{6.49}    \\
	                                         & Local Role MI               & 0.22      & 0.23                           & \textcolor{blue}{-7.69}   &                                    & 0.22      & \textcolor{blue}{-3.88}  \\
	                                         & Local Role Entropy         & 0.87      & 0.88                           & \textcolor{blue}{-4.73}   &                                    & 0.88      & \textcolor{blue}{-6.32}  \\
	                                         & Node Role Entropy           & 0.05      & 0.07                           & \textcolor{blue}{-19.78}  &                                    & 0.05      & 0.00                     \\
	                                         & Weight. Degree Entropy      & 13.89     & 13.81                          & \textcolor{red}{14.97}    &                                    & 13.84     & \textcolor{red}{8.71}    \\
	                                         & Weight. Eigenvector Entropy & 7.95      & 8.19                           & -0.38                     &                                    & 7.28      & 1.24                     \\
	                                         & Weight. Pagerank Entropy    & 14.28     & 14.23                          & \textcolor{red}{18.13}    &                                    & 14.24     & \textcolor{red}{13.88}   \\
	\bottomrule
\end{tabular}
                }
        \label{tab:ensemble_significance}
    \end{table}

    A full summary of the of the analysis across all features and datasets is given in \Cref{app:ensemble}.

\section{Discussion}
\label{sec:discussion}
    
    Many of the seminal results in network science were obtained using dyadic, unweighted networks without metadata -- perhaps the minimal model of a complex system. 
    As the quantity and richness of networked data have grown, so too has the need to incorporate higher-order interactions and heterogeneous node and edge properties into our models.
    Toward this program, we have introduced annotated hypergraphs. 
    Annotated hypergraphs may be viewed as a natural extension of directed graphs for modeling heterogeneous, polyadic interactions. 
    Because the setting of annotated hypergraphs is highly general, they allow the data scientist to flexibly incorporate their assumptions about how role information should feature into downstream analysis. 
    
    We have also made contributions to the inferential and exploratory analysis of annotated hypergraphs. 
    First, we formulated a role-aware configuration null model. 
    This model can be used to assess whether an observed structural feature in an annotated hypergraph can be explained purely through information about role-dependent edge and node incidence. 
    Features that cannot may be reasonably attributed to higher-order mechanisms, which may then be investigated.
    Second, we provide a small suite of analytical tools that generalize many common methods for studying dyadic networks, including centrality, assortativity, and modularity scores. 
    We show how each of these may be used in role-aware ways to highlight diverse features of the data. 

    Annotated hypergraphs admit multiple avenues of future work. 
    There are opportunities to define and study diffusion, spreading, and opinion dynamics on annotated hypergraphs, using models that account for heterogeneous, polyadic interactions.
    For example, our weighted projection scheme does not incorporate any explicit accounting of edge dimensions. 
    It may be of interest to define a role-dependent simple random walk along the hypergraph. 
    These structures implicitly normalize edge dimensions, implying that a node who received an email along with ten others is in some sense less important than one who received a private communication. 
    This walk could then be used to define alternative measures of centrality and modularity. 

    Another direction of future development concerns the structure of hyperedges. 
    By assigning roles within each hyperedge, we impose a certain model of how the interaction marked by the edge takes place. 
    Further structural assumptions are possible.
    In some cases it may be useful to assume, for example, that the hyperedge itself contains a small network between the nodes. 
    The role of the hyperedge in this case is to serve as a single entity housing a network motif \cite{alon2007network}. 
    A null model over such structures would allow for the sampling of random graphs with control over the participation of nodes in various microscale graph structures. 
    While such a model would be substantially more complex, the emerging importance of network motifs \cite{benson2016higher} and network-of-networks modeling (e.g. \cite{kenett2015networks}) may provide sufficient impetus to pursue it. 
    
    Finally, an important feature of many polyadic data sets is that interactions are temporally localized. 
    The incorporation of temporal information into models of hypergraphs and annotated hypergraphs is of substantial importance for modeling realistic dynamics on network substrates. 
    One route may be to generalize temporal event graphs \cite{mellor2018event} for rich, polyadic data. 
    Such a generalization, along with the development of associated metrics, would be of substantial theoretical and practical interest. 

    This work is a contribution to the project of integrating progressively more complex information into network data science. 
    We foresee that this program will become increasingly important as rich, relational data sets become more readily available.

\section*{Declarations}

\subsubsection*{Availability of data and materials}
    The datasets generated and analysed during the current study are  \href{https://andrewmellor.co.uk/data/}{available from the authors webpage}.
	A Python package for the analysis of annotated hypergraphs is \href{https://github.com/PhilChodrow/annotated_hypergraphs}{available on Github}.
    This includes the analysis to replicate the figures from this study.
    
\subsubsection*{Competing interests}
The authors declare that they have no competing interests.

\subsubsection*{Funding}
PC is supported by the US National Science Foundation Graduate Research Fellowship under grant number 1122374. 
AM is funded by the Oxford-Emirates Data Science Laboratory.

\subsubsection*{Author contributions}
PC and AM contributed equally to the conception of the research, the data analysis, the technical analysis, and the writing of the manuscript.

\subsubsection*{Acknowledgments}
    
    The idea for this research was born at the 2019 SIAM Workshop on Network Science, organized by Nina Fefferman and Peter Mucha. 
    PC also thanks the University of Oxford for financial support of a brief collaborative visit.  

\bibliographystyle{abbrv}
\bibliography{references}

\begin{thebibliography}{10}

\bibitem{alon2007network}
U.~Alon.
\newblock Network motifs: {{Theory}} and experimental approaches.
\newblock {\em Nature Reviews Genetics}, 8(6):450--461, 2007.

\bibitem{Angel2016}
O.~Angel, R.~van~der Hofstad, and C.~Holmgren.
\newblock {Limit laws for self-loops and multiple edges in the configuration
  model}.
\newblock {\em arXiv:1603.07172}, pages 1--19, 2016.

\bibitem{battiston2014structural}
F.~Battiston, V.~Nicosia, and V.~Latora.
\newblock Structural measures for multiplex networks.
\newblock {\em Physical Review E}, 89(3):032804, 2014.

\bibitem{benson2019three}
A.~R. Benson.
\newblock Three hypergraph eigenvector centralities.
\newblock {\em SIAM Journal on Mathematics of Data Science}, 1(2):293--312,
  2019.

\bibitem{Benson2018}
A.~R. Benson, R.~Abebe, M.~T. Schaub, A.~Jadbabaie, and J.~Kleinberg.
\newblock {Simplicial closure and higher-order link prediction}.
\newblock {\em Proceedings of the National Academy of Sciences},
  115(48):11221--11230, 2018.

\bibitem{benson2016higher}
A.~R. Benson, D.~F. Gleich, and J.~Leskovec.
\newblock Higher-order organization of complex networks.
\newblock {\em Science}, 353(6295):163--166, 2016.

\bibitem{berge1984hypergraphs}
C.~Berge.
\newblock {\em Hypergraphs: Combinatorics of Finite Sets}, volume~45.
\newblock {Elsevier}, 1984.

\bibitem{Bollobas1980}
B.~Bollob{\'{a}}s.
\newblock {A probabilistic proof of an asymptotic formula for the number of
  labelled regular graphs}.
\newblock {\em European Journal of Combinatorics}, 1(4):311--316, 1980.

\bibitem{carlsson2009topology}
G.~Carlsson.
\newblock Topology and data.
\newblock {\em Bulletin of the American Mathematical Society}, 46(2):255--308,
  2009.

\bibitem{chodrow2019configuration}
P.~S. Chodrow.
\newblock Configuration {{Models}} of {{Random Hypergraphs}} and their
  {{Applications}}.
\newblock {\em arXiv:1902.09302 [physics, stat]}, Feb. 2019.

\bibitem{chodrow2019moments}
P.~S. Chodrow.
\newblock Moments of uniform random multigraphs with fixed degree sequences.
\newblock {\em arXiv preprint arXiv:1909.09037}, 2019.

\bibitem{de2019social}
G.~F. de~Arruda, G.~Petri, and Y.~Moreno.
\newblock Social contagion models on hypergraphs.
\newblock {\em arXiv preprint arXiv:1909.11154}, 2019.

\bibitem{Fosdick2018}
B.~K. Fosdick, D.~B. Larremore, J.~Nishimura, and J.~Ugander.
\newblock Configuring random graph models with fixed degree sequences.
\newblock {\em SIAM Review}, 60(2):315--355, 2018.

\bibitem{gallo1993directed}
G.~Gallo, G.~Longo, S.~Pallottino, and S.~Nguyen.
\newblock Directed hypergraphs and applications.
\newblock {\em Discrete applied mathematics}, 42(2-3):177--201, 1993.

\bibitem{gallo1998directed}
G.~Gallo and M.~G. Scutella.
\newblock Directed hypergraphs as a modelling paradigm.
\newblock {\em Rivista di matematica per le scienze economiche e sociali},
  21(1-2):97--123, 1998.

\bibitem{ghoshal2009random}
G.~Ghoshal, V.~Zlati{\'c}, G.~Caldarelli, and M.~Newman.
\newblock Random hypergraphs and their applications.
\newblock {\em Physical Review E}, 79(6):066118, 2009.

\bibitem{gomez2013diffusion}
S.~Gomez, A.~Diaz-Guilera, J.~Gomez-Gardenes, C.~J. Perez-Vicente, Y.~Moreno,
  and A.~Arenas.
\newblock Diffusion dynamics on multiplex networks.
\newblock {\em Physical review letters}, 110(2):028701, 2013.

\bibitem{greenhill2011polynomial}
C.~Greenhill.
\newblock A polynomial bound on the mixing time of a markov chain for sampling
  regular directed graphs.
\newblock {\em The Electronic Journal of Combinatorics}, 18(1):234, 2011.

\bibitem{greenhill2014switch}
C.~Greenhill.
\newblock The switch markov chain for sampling irregular graphs.
\newblock In {\em Proceedings of the Twenty-Sixth Annual ACM-SIAM Symposium on
  Discrete Algorithms}, pages 1564--1572. SIAM, 2014.

\bibitem{greening2015higher}
B.~R. Greening~Jr, N.~Pinter-Wollman, and N.~H. Fefferman.
\newblock Higher-order interactions: understanding the knowledge capacity of
  social groups using simplicial sets.
\newblock {\em Current Zoology}, 61(1):114--127, 2015.

\bibitem{heath2009multimodal}
L.~S. Heath and A.~A. Sioson.
\newblock Multimodal networks: Structure and operations.
\newblock {\em IEEE/ACM Transactions on Computational Biology and
  Bioinformatics (TCBB)}, 6(2):321--332, 2009.

\bibitem{henderson2012rolx}
K.~Henderson, B.~Gallagher, T.~{Eliassi-Rad}, H.~Tong, S.~Basu, L.~Akoglu,
  D.~Koutra, C.~Faloutsos, and L.~Li.
\newblock Rolx: Structural role extraction \& mining in large graphs.
\newblock In {\em Proceedings of the 18th {{ACM SIGKDD}} International
  Conference on {{Knowledge}} Discovery and Data Mining}, pages 1231--1239.
  {ACM}, 2012.

\bibitem{Kami2018}
B.~Kaminski, V.~Poulin, P.~Pralat, P.~Szufel, and F.~Theberge.
\newblock {Clustering via hypergraph modularity}.
\newblock {\em arXiv:1810.04816}, pages 1--17, 2018.

\bibitem{kenett2015networks}
D.~Y. Kenett, M.~Perc, and S.~Boccaletti.
\newblock Networks of networks--an introduction.
\newblock {\em Chaos, Solitons \& Fractals}, 80:1--6, 2015.

\bibitem{klamt2009hypergraphs}
S.~Klamt, U.-U. Haus, and F.~Theis.
\newblock Hypergraphs and cellular networks.
\newblock {\em PLoS computational biology}, 5(5):e1000385, 2009.

\bibitem{klimt2004introducing}
B.~Klimt and Y.~Yang.
\newblock Introducing the enron corpus.
\newblock In {\em CEAS}, 2004.

\bibitem{kovanen2013temporal}
L.~Kovanen, K.~Kaski, J.~Kert{\'e}sz, and J.~Saram{\"a}ki.
\newblock Temporal motifs reveal homophily, gender-specific patterns, and group
  talk in call sequences.
\newblock {\em Proceedings of the National Academy of Sciences},
  110(45):18070--18075, 2013.

\bibitem{leicht2008community}
E.~A. Leicht and M.~E. Newman.
\newblock Community structure in directed networks.
\newblock {\em Physical review letters}, 100(11):118703, 2008.

\bibitem{marcotte1998hyperpath}
P.~Marcotte and S.~Nguyen.
\newblock Hyperpath formulations of traffic assignment problems.
\newblock In {\em Equilibrium and advanced transportation modelling}, pages
  175--200. Springer, 1998.

\bibitem{mcmorris1994triangulating}
F.~R. McMorris, T.~J. Warnow, and T.~Wimer.
\newblock Triangulating vertex-colored graphs.
\newblock {\em SIAM Journal on Discrete Mathematics}, 7(2):296--306, 1994.

\bibitem{mellor2018event}
A.~Mellor.
\newblock Event {{Graphs}}: {{Advances}} and {{Applications}} of
  {{Second}}-order {{Time}}-unfolded {{Temporal Network Models}}.
\newblock {\em arXiv preprint arXiv:1809.03457}, 2018.

\bibitem{Molloy1998}
M.~Molloy and B.~Reed.
\newblock {The size of the giant component of a random graph with a given
  degree sequence}.
\newblock {\em Combinatorics, Probability, and Computing}, 7(3):295--305, 1998.

\bibitem{mucha2010community}
P.~J. Mucha, T.~Richardson, K.~Macon, M.~A. Porter, and J.-P. Onnela.
\newblock Community structure in time-dependent, multiscale, and multiplex
  networks.
\newblock {\em science}, 328(5980):876--878, 2010.

\bibitem{newman2010networks}
M.~Newman.
\newblock {\em Networks: {{An Introduction}}}.
\newblock {Oxford University Press}, 2010.

\bibitem{page1999pagerank}
L.~Page, S.~Brin, R.~Motwani, and T.~Winograd.
\newblock The {{PageRank Citation Ranking}}: {{Bringing Order}} to the {{Web}}.
\newblock Technical {{Report}} 1999-66, {Stanford InfoLab}, 1999.

\bibitem{peel2017ground}
L.~Peel, D.~B. Larremore, and A.~Clauset.
\newblock The ground truth about metadata and community detection in networks.
\newblock {\em Science advances}, 3(5):e1602548, 2017.

\bibitem{rotabi2017tracing}
R.~Rotabi, C.~Danescu-Niculescu-Mizil, and J.~Kleinberg.
\newblock Tracing the use of practices through networks of collaboration.
\newblock In {\em Eleventh International AAAI Conference on Web and Social
  Media}, 2017.

\bibitem{xie2016spectral}
J.~Xie and L.~Qi.
\newblock Spectral directed hypergraph theory via tensors.
\newblock {\em Linear and Multilinear Algebra}, 64(4):780--794, 2016.

\bibitem{Young2017}
J.~G. Young, G.~Petri, F.~Vaccarino, and A.~Patania.
\newblock {Construction of and efficient sampling from the simplicial
  configuration model}.
\newblock {\em Physical Review E}, 96(3):1--6, 2017.

\bibitem{zhang2015multiway}
X.~Zhang and M.~E. Newman.
\newblock Multiway spectral community detection in networks.
\newblock {\em Physical Review E}, 92(5):052808, 2015.

\bibitem{zhou2007learning}
D.~Zhou, J.~Huang, and B.~Sch{\"o}lkopf.
\newblock Learning with hypergraphs: Clustering, classification, and embedding.
\newblock In {\em Advances in neural information processing systems}, pages
  1601--1608, 2007.

\end{thebibliography}

\appendix

\section{Data and Software}
\label{app:data}

\subsection{Data Sets and Choice of Role-Interaction Matrix}

\subsubsection{Enron Email}

	This consists of the core Enron emails from the archived Enron email database \cite{klimt2004introducing}.
	Core email addresses are those with a valid Enron address.
	All other emails have been omitted.
	Here the node roles are \emph{from}, \emph{to}, and \emph{cc} which capture the various fields in a typical email header (in this case bcc has been merged with cc).
	Note that a node may appear in an edge twice under multiple roles, for example sending a message to oneself.
    
    For this hypergraph we choose the role-interaction matrix to be
	\begin{align*}
        \mathbf{R} = \kbordermatrix{
            & \text{from} & \text{to} & \text{cc} \\ 
            \text{from} &0 & 1 & 0.25\\ 
            \text{to} &0 & 0 & 0\\ 
            \text{cc} &0 & 0 & 0
        }.
    \end{align*}
    
    This reflects that information can only be transmitted from the sender.
    Furthermore this reflects the assumption that information is less likely to be transmitted (or not fully transmitted) to those who are ``cc'd.''

\subsubsection{Scopus Multilayer Literature}

	This data consists of academic literature surrounding `multilayer networks,' collected using Scopus.
	Specifically we choose all references from the following three reviews and books:
	\begin{enumerate}
	    \item Kivelä et al. "Multilayer networks." Journal of Complex Networks 2.3 (2014): 203-271.
	    \item Bianconi. Multilayer networks: structure and function. Oxford University Press, 2018.
	    \item Boccaletti et al. "The structure and dynamics of multilayer networks." Physics Reports 544.1 (2014): 1-122.
	\end{enumerate}
	Authors are assigned roles for each article dependent on their order in the list of authors.
	Although practices vary across disciplines and institutions, here we distinguish between first, middle, and last authors.
	When there are fewer than three authors then the role is assigned as first for a single author, and first and last for a pair of authors (regardless of any note of equal contribution).

    For this hypergraph we choose the role-interaction matrix to be
	\begin{align*}
        \mathbf{R} = \kbordermatrix{
            & \text{first} & \text{middle} & \text{last} \\ 
            \text{first} &0 & 1 & 0.5\\ 
            \text{middle} &0.2 & 0.2 & 0.2\\ 
            \text{last} &1 & 0.25 & 0
        }.
    \end{align*}
    We make the modelling assumption that the first author is the most knowledgeable and therefore able to spread information to other authors.
    The last author is often an advisor who can spread information to the first author, while the middle authors can weakly diseminate information between everyone.

\subsubsection{MovieLens Actor Credits}

	This data contains a list of credits from the MovieLens data collection.
	We consider the top five billed actors from a collection of [] movies.
	Actor roles are distinguished by being the top-billed actor, or in the remaining cast. 

    For this hypergraph we choose the role-interaction matrix to be
	\begin{align*}
        \mathbf{R} = \kbordermatrix{
            & \text{top} & \text{rest}  \\ 
            \text{top} &0 & 1 \\ 
            \text{rest} &0.25 & 0.25 \\ 
        }.
    \end{align*}
    The top billed actor is assumed to be the most diffusive, potentially spreading fame and influence.
    The lower billed cast have a smaller diffusive rate.

\subsubsection{Stack Overflow Threads}

	This data contains a list of Stack Overflow question threads which achieved a score greater than 25 between 1st January 2017 to 1st January 2019.
	The score is calculated and reflects the quality of the question both in terms of its pertinence and its presentation.
	Here edges reflect questions threads where users can be in three roles.
	These are the question setter, the question answerers, and the best answerer (chosen by the question setter as the accepted answer).
    
    For this hypergraph we choose the role-interaction matrix to be    
	\begin{align*}
        \mathbf{R} = \kbordermatrix{
            & \text{setter} & \text{answerer} & \text{accepted} \\ 
            \text{setter} &0 & 0.1 & 0.1\\ 
            \text{answerer} &0.3 & 0.3 & 0.3\\ 
            \text{accepted} &1 & 0.5 & 0
        }.
    \end{align*}
    Here, the question setter may disseminate some information (either about the question, or the topic).
    The question answerers may share information uniformly across all roles.
    The node whose question is answered transfers the most information to the question setter since the setter has chosen this response to adopt. 
    This accepted answer may also benefit the other answerers with less useful answers.

\subsubsection{Math Overflow Threads}

	This data is in the identical format to the Stack Overflow threads, however questions threads are now on the topic of mathematical research as opposed to general programming.

\subsubsection{Twitter Keyword Sample}

    This data contains a list of messages posted on the social media platform Twitter over a 24 hour period.
    All these messages contained a particular keyword relating to the aviation industry.
    Each edge corresponds to a message (or \emph{tweet}).
    Nodes participating in a message can have four possible roles, a sender, a receiver, a retweeter, and the retweeted\footnote{See \href{https://help.twitter.com/en/twitter-guide}{https://help.twitter.com/en/twitter-guide} for details of each role.}.
    
    For this hypergraph we choose the role-interaction matrix to be
	\begin{align*}
        \mathbf{R} = \kbordermatrix{
            & \text{sender} & \text{receiver} & \text{retweeter} & \text{retweeted} \\ 
            \text{sender} &0 & 1 & 0 & 0\\ 
            \text{receiver} &0 & 0 & 0 & 0\\ 
            \text{retweeter} &0 & 0.25 & 0 & 0\\
            \text{retweeted} &0 & 0.25 & 1 & 0\\          
        }.
    \end{align*}
    We assume that the sender transmits information to the receivers in a directed fashion (information travelling only one way).
    A retweeted node can transmit information to the node that retweets it and so to any receivers who are also included in the message.

\section{Ensemble Study}
\label{app:ensemble}

In \Cref{fig:ensemble_full} we present the all results from the ensemble study.

\begin{figure}
    \centering
    \includegraphics[height=6in]{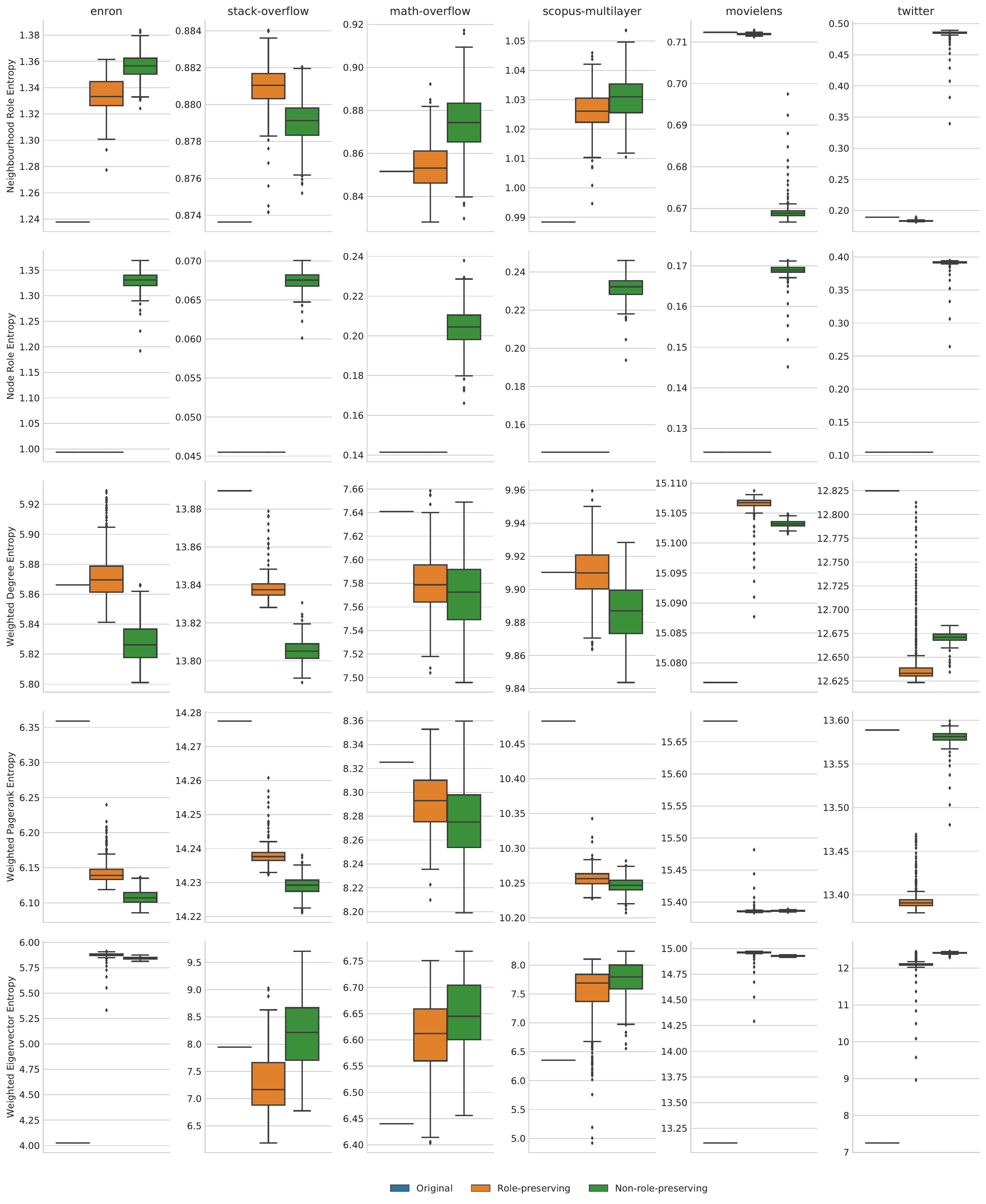}
    \caption{Full results.}
    \label{fig:ensemble_full}
\end{figure}

\end{document}